\RequirePackage{fix-cm}
\documentclass[smallextended]{svjour3}       
\smartqed  
\usepackage{graphicx}
\usepackage{microtype}

\usepackage[all]{xy}

\usepackage{epsfig}
\usepackage{amssymb}
\usepackage{mathtools}
\usepackage{physics}
\usepackage{enumerate}
\usepackage{float}
\usepackage{authblk}
\usepackage{scalerel}

\usepackage[colorlinks,linktocpage]{hyperref}

\newcommand{\mysymbol}[1]{{\mbox{\raisebox{-0.3em}{\epsfysize=1.2em\epsfbox{#1}}}}}

\newcommand{\tileC}{\mysymbol{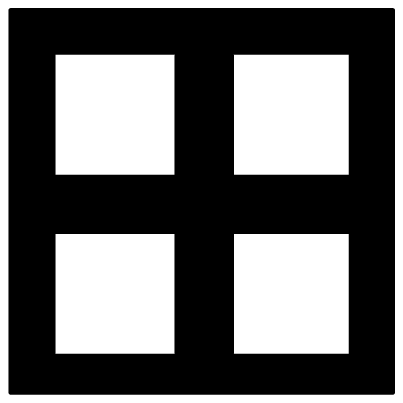}}       
\newcommand{\tileL}{\mysymbol{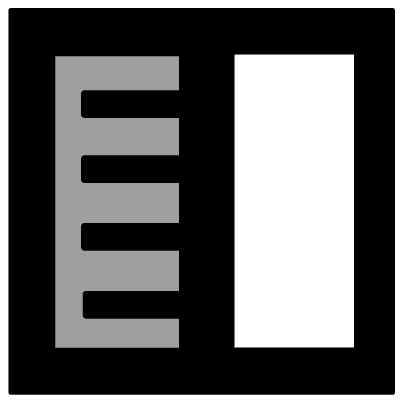}}      
\newcommand{\tileN}{\mysymbol{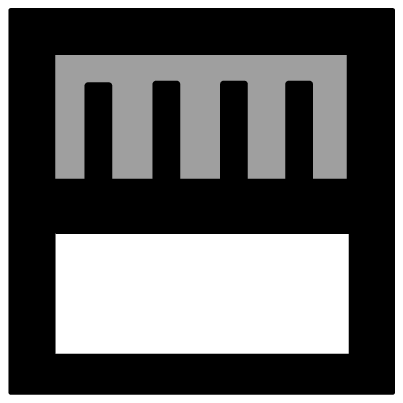}}       
\newcommand{\tileS}{\mysymbol{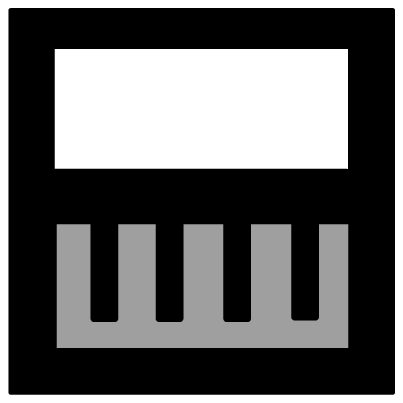}}       
\newcommand{\tileR}{\mysymbol{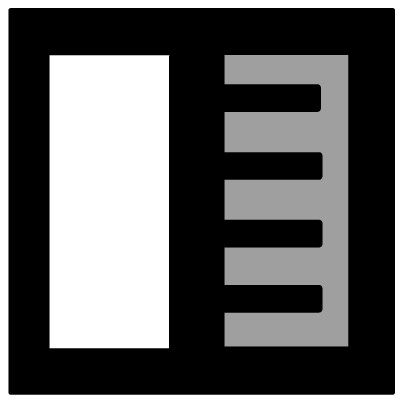}}       
\newcommand{\tileSR}{\mysymbol{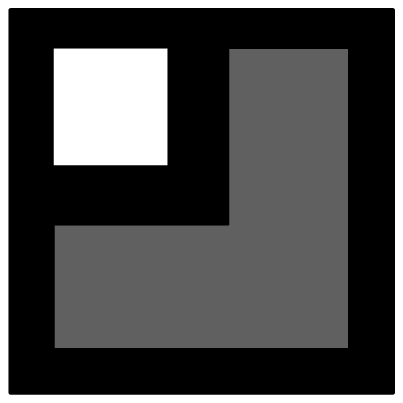}}
\newcommand{\tileSL}{\mysymbol{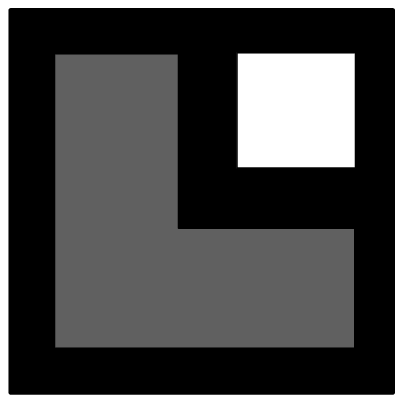}}
\newcommand{\tileNR}{\mysymbol{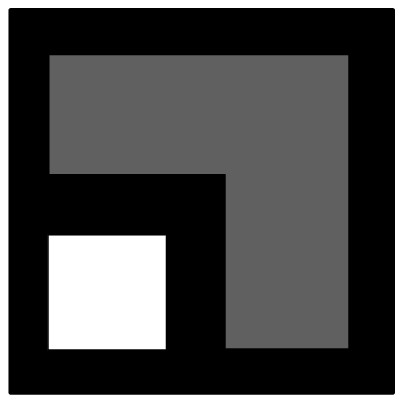}}
\newcommand{\tileNL}{\mysymbol{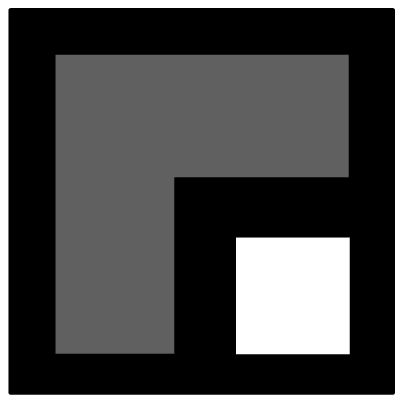}}
\newcommand{\tileW}{\mysymbol{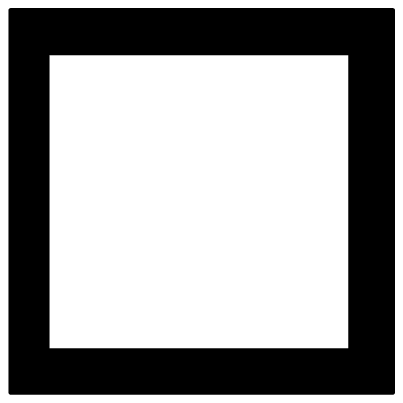}}
\newcommand{\tileB}{\mysymbol{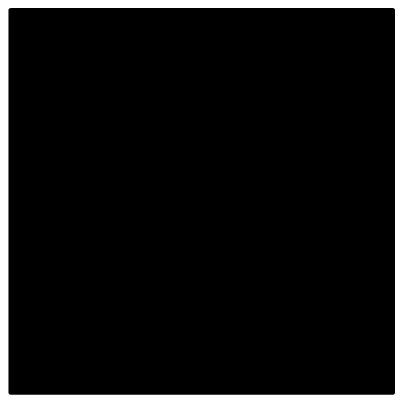}}
\newcommand{\tileV}{\mysymbol{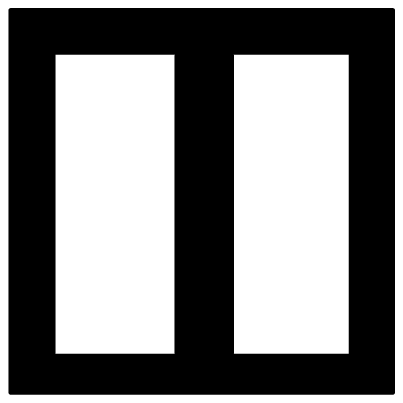}}
\newcommand{\tileH}{\mysymbol{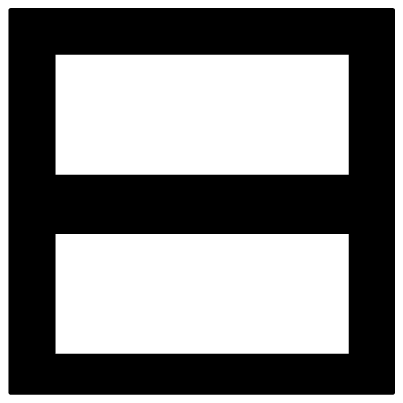}}

\renewcommand{\paragraph}[1]{\textbf{\textit{#1}}}

\journalname{Journal of Statistical Physics}
\begin{document}

\title{Translationally invariant universal classical Hamiltonians}


\author{Tamara Kohler         \and
        Toby Cubitt 
}


\institute{T. Kohler \at
              Department of Computer Science, University College London, London, WC1E 6BT, UK \\
              \email{tamara.kohler.16@ucl.ac.uk}           
           \and
           T. Cubitt \at
              Department of Computer Science, University College London, London, WC1E 6BT, UK \\
              \email{t.cubitt@ucl.ac.uk}
}

\date{Received: date / Accepted: date}

\maketitle

\begin{abstract}
Spin models are widely studied in the natural sciences, from investigating magnetic materials in condensed matter physics to studying neural networks. Previous work has demonstrated that there exist simple classical spin models that are universal: they can replicate -- in a precise and rigorous sense -- the complete physics of any other classical spin model, to any desired accuracy. However, all previously known universal models break translational invariance. In this paper we show that there exist translationally invariant universal models. Our main result is an explicit construction of a translationally invariant, 2D, nearest-neighbour, universal classical Hamiltonian with a single free parameter. The proof draws on techniques from theoretical computer science, in particular recent complexity theoretic results on tiling problems. Our results imply that there exists a single Hamiltonian which encompasses all classical spin physics, just by tuning a single parameter and varying the size of the lattice. We also prove that our construction is optimal in terms of the number of parameters in the Hamiltonian; there cannot exist a translationally invariant universal Hamiltonian with only the lattice size as a parameter.

\end{abstract}

\section{Introduction}
\label{intro}
Classical spin models are ubiquitous in statistical physics. They were first introduced to study magnetism in condensed matter physics \cite{Baxter:1982}, and have since been employed to study features of interacting systems in diverse areas of the natural and social sciences.


Recent work has demonstrated that there exist families of spin Hamiltonians that are universal, in the sense that they can replicate the physics of all other classical spin Hamiltonians \cite{Cubitt:2016}. One drawback to the universality result derived in \cite{Cubitt:2016} is that all the universal models found in the paper break translational invariance. In this paper we improve on the result by demonstrating that there exists a family of translationally invariant classical Hamiltonians which is universal.


The work on universal Hamiltonians arose in part due to previous work on the completeness of the partition functions of a set of classical spin models \cite{Nest:2008}, \cite{Karimipour:2012}, \cite{Cuevas:2009}, \cite{Duer:2009}, \cite{Zarei:2012}, \cite{Xu:2011}, where a model is said to be ``Hamiltonian complete'' if its partition function can replicate (up to a constant multiplicative factor) the partition function of any other model.

The requirements for a ``universal model'' in \cite{Cubitt:2016} are more demanding: a model is universal if for any classical Hamiltonian $H'$ there exists a Hamiltonian from the model $H$ that can simulate $H'$,
where ``simulate'' means reproducing not only the partition function but also the energy levels and spin configurations.
(See Definition \ref{simulation_definition} or \cite{Cubitt:2016} for the mathematically rigorous definition.)

The main technical result in \cite{Cubitt:2016} is that a spin model is a universal model if and only if it is closed, and its ground state energy problem admits a polynomial-time faithful reduction from SAT.\footnote{Any Hamiltonian whose ground state energy problem is NP-complete will admit a polynomial time reduction from SAT, so NP-completeness of the ground state energy problem is a necessary (although not sufficient) condition for universality.} The authors then show that the 2D Ising model with fields meets these criteria, hence is a universal classical Hamiltonian. They also prove universality of a number of other simple spin models, including the 3D Ising model and the Potts model. However all the models shown to be universal in \cite{Cubitt:2016} require the ability to tune individual interaction strengths in the Hamiltonian, raising the question of whether the universality is a consequence of breaking translational invariance in this way.

In \cite{Gottesman:2009} it was shown that the translationally invariant tiling problem is NEXP-complete, demonstrating that translational invariance is not a barrier to complexity. In this paper we go further, and show that translational invariance is not a barrier to universality. Our main result is that there exists a family of translationally invariant Hamiltonians, parameterised by a single parameter $\tilde{\Delta}$, which is universal. This translationally invariant universal model is defined on a 2D square lattice of spins with nearest-neighbour interactions. The Hamiltonian can be written in the form:
\begin{equation}
H(\sigma) = \tilde{\Delta}H_1 + H_2
\end{equation}
where $H_1$ and $H_2$ are fixed, translationally invariant, nearest-neighbour Hamiltonians. By varying the size of the lattice that the Hamiltonian is acting on, and tuning the $\tilde{\Delta}$ parameter in the construction, this family of Hamiltonians can replicate all classical spin physics in the strong sense of \cite{Cubitt:2016}. 

\section{Main results}\label{main_results}

In this section we give the main results and a high-level overview of the construction. See Sections \ref{NEXP-complete},\ref{universality} for full technical details of the constructions and proofs.

\subsection{Universal model described by two parameters}
Our main result constructs a single-parameter translationally invariant universal model. But we also prove universality of a two-parameter model, where the scaling of the parameters in the two-parameter model is better in some cases than the scaling in the one-parameter model. We will consider the simpler two-parameter result first, as the one-parameter construction builds on this.

\begin{theorem}
There exist fixed two-body interactions $h_{3,\mathrm{vert}}$, $h_{3,\mathrm{horiz}}$, $h_{4,\mathrm{vert}}$, $h_{4,\mathrm{horiz}}$ such that the family of translationally-invariant Hamiltonians on a 2D square lattice of size $N$ with $h_3,h_4$ as nearest-neighbour interactions:
\begin{equation}\label{h}
  H = \Delta H_3 +  \alpha H_4 \quad \text{with} \quad
  H_x = \sum_{\langle i,j\rangle_\mathrm{row}} h_{x,\mathrm{horiz}}^{(i,j)}
        + \sum_{\langle i,j\rangle_\mathrm{col}} h_{x,\mathrm{vert}}^{(i,j)}
\end{equation}
is an efficient universal model, where $\Delta$, $\alpha$ are parameters of the Hamiltonian, and the sums are over adjacent sites along the rows and columns, respectively.
\end{theorem}
\noindent (The precise meaning of efficiency in the translationally invariant case is given in Section \ref{proof}.)

The universality construction relies on the fact that it is possible to encode the evolution of a Turing machine into the ground state of a nearest neighbour, translationally invariant spin Hamiltonian on a 2D lattice. A Turing machine is a model of computation in which a head reads and writes symbols from some finite alphabet on a tape and moves left or right, following a finite set of transition rules. The transition rules are the same regardless of the location of the head along the tape. Full details of how a Turing machine can be encoded into a Hamiltonian are given in Section \ref{gottesman_details}.

In Eq. \ref{h}, the $H_3$ term is a Hamiltonian which encodes a Turing machine. The $H_3$ Turing machine reads in a description of the Hamiltonian to be simulated, $H'$, and a spin configuration, $\sigma$, and computes the energy of $\sigma$ with respect to $H'$. The output of the $H_3$ Turing machine is a certain number of `flag' spin states, which pick up energy from the $H_4$ term in Eq. \ref{h}. In this way $H$ reproduces the energy levels of $H'$ below some energy cut-off $\Delta$.

The $H_3$ and $H_4$ Hamiltonians in Eq. \ref{h} are both fixed - which raises the question of where the information about $H'$ is encoded? The basic idea is to encode information about $H'$ into the binary expansion of the size of the lattice, $N$ (the same trick was used in \cite{Gottesman:2009}). The computation encoded into $H_3$ begins by extracting the binary representation of $N$ from the grid size, by incrementing a binary counter for a number of time-steps equal to the grid size. The binary representation of $N$ is then used as input for the main computation. (Full details of the construction are given in Section \ref{proof}.)

\subsection{Universal model described by one parameter}
We can now turn to the one parameter universal model.

\begin{theorem}
There exist fixed two-body interactions $h_{1,\mathrm{vert}}$, $h_{1,\mathrm{horiz}}$, $h_{2,\mathrm{vert}}$, $h_{2,\mathrm{horiz}}$ such that the family of translationally-invariant Hamiltonians on a 2D square lattice of size $N$ with $h_1,h_2$ as nearest-neighbour interactions:
\begin{equation}
  H = \tilde{\Delta} H_1 +  H_2 \quad \text{with} \quad
  H_x = \sum_{\langle i,j\rangle_\mathrm{row}} h_{x,\mathrm{horiz}}^{(i,j)}
        + \sum_{\langle i,j\rangle_\mathrm{col}} h_{x,\mathrm{vert}}^{(i,j)}
\end{equation}
is an efficient universal model, where $\tilde{\Delta}$ is a parameter of the Hamiltonian, and the sums are over adjacent sites along the rows and columns, respectively.
\end{theorem}

The one-parameter construction works in much the same as the two-parameter construction. The $H_1$ term is again a Hamiltonian which encodes a Turing machine, and the $H_2$ term is used to produce the required energy. By using a simple lemma about irrational numbers we can remove the need for the $\alpha$ parameter in the construction. The price we pay is that the $H_2$ Hamiltonian now contains negative energy terms. These result in a larger scaling of the $\tilde{\Delta}$ parameter in the Hamiltonian (as compared with the $\Delta$ in the two-parameter construction), to deal with the fact that invalid spin configurations may pick up negative energy bonuses. Full details are given in Section \ref{two-parameters}.

Our final result concerns the impossibility of constructing zero-parameter universal models, where the only thing that can vary is the number of spins on which the Hamiltonian acts:
\begin{theorem}
It is not possible to construct a translationally invariant universal model whose only parameter is the number of spins on which the Hamiltonian acts.
\end{theorem}
\noindent Hence our one-parameter model is optimal in terms of the number of parameters required for a universal Hamiltonian.

The existence of translationally invariant universal models implies that every hardness result known about general classical Hamiltonians can now be extended to translationally invariant classical Hamiltonians, where the results are shifted up in time-complexity by an exponential factor due to the way problem instances are encoded. (Details of this are given in Section \ref{proof}.)


\section{Preliminaries}
\subsection{Classical spin Hamiltonians}
Discrete spins (also known as Ising spins) are variables which can take on values in some set of states $\mathcal{S}$. Given a set of spins $\{ \sigma_i\}$ for $i \in \{1,...,n\}$, a spin configuration assigns a state from $\mathcal{S}$ to each spin $\sigma_i$. A classical spin Hamiltonian, $H$, is a function, $H: \mathcal{S}^{\times n} \rightarrow \mathbb{R}$ which specifies the energy $H(\sigma)$ of each spin configuration $\sigma = \sigma_1\sigma_2\dots\sigma_n \in \mathcal{S}^{\times n}$.
We can also consider continuous spins, represented as unit vectors $\sigma_i \in \mathbb{S}^D$, where $\mathbb{S}^D$ is the $D$-dimensional unit sphere.

We refer to families of related spin Hamiltonians as ``spin models''. In this work, the spin models we consider will be translationally invariant Hamiltonians on a 2D square lattice with nearest-neighbour interactions, with some small number (1 or 2) of global parameters. Different Hamiltonians from the same model therefore differ only in the size of the lattice and the values of the parameters.

\subsection{$k$-local Hamiltonians}
Throughout our discussion of universal models we will assume that the Hamiltonians we are simulating are $k$-local:
\begin{equation}
H = \sum_{i} H^{(i)}
\end{equation}
where each $H^{(i)}$ acts non-trivially on at most $k$ spins. Note that there is no assumption that the $H^{(i)}$ are geometrically local, we only require that the number of spins involved in every interaction is upper bounded by $k$.

While at first restricting ourselves to $k$-local Hamiltonians may appear restrictive, we place no restriction on how large $k$ is allowed to be, so a global Hamiltonian on $n$ spins is just a special case of a $k$-local Hamiltonian for which $n=k$. Phrasing everything in terms of $k$-local Hamiltonians allows us to derive efficiency results at the same time as universality results.

The choice of defining efficiency in terms of $k$-local Hamiltonians is convenient, and covers the most important case. But it is more restrictive than is really required. When we talk about simulating a Hamiltonian efficiently, the most general requirement is that the number of parameters specifying the simulator Hamiltonian scales at most polynomially in terms of the number of parameters required to describe in the original system. For global Hamiltonians with no structure this requirement becomes a triviality, as the number of bits of information needed to describe a global Hamiltonian with no structure is anyway exponential in the number of spins. However, one could construct non-local Hamiltonians with a structure which ensures they have efficiently computable energy levels.
In the interests of simplicity of exposition, we will only give efficient explicit constructions for $k$-local Hamiltonians. But our constructions can easily be generalised to give efficient simulations of all Hamiltonians for which the energy of a given spin configuration can be computed in time $O(e^n)$, where $n$ is the number of spins in the system. The energy of $k$-local Hamiltonians can be computed in time $\text{poly}(n)$, so $k$-local Hamiltonians are well within this bound.

\subsection{Simulation and universality}
A rigorous definition of what it means for one classical Hamiltonian to simulate another was formulated in \cite{Cubitt:2016} (\cite{Cubitt:2017} extends this defintion to the quantum case):

\begin{definition}[Hamiltonian simulation (definition from main text of \cite{Cubitt:2016}\footnote{We have made the bound on the error in the partition function more precise than that in \cite{Cubitt:2016}.})]  \label{simulation_definition}
We say that a spin model with spin degrees of freedom $\sigma=\sigma_1,\sigma_2,\dots$ can \emph{simulate} $H'$ if it satisfies all three of the following:
\begin{enumerate}
\item For any $\Delta > \max_{\sigma'} H'(\sigma')$ and any $0 < \delta < 1$, there exists a Hamiltonian $H$ in the model whose low-lying energy levels $E_\sigma=H(\sigma)<\Delta$ approximate the energy levels $E'_{\sigma'}=H'(\sigma')$ of $H'$ to within additive error $\delta$.
 \label{part:eigenvalues}
\item For every spin $\sigma'_i$ in $H'$, there exists a fixed subset $P_i$ of the spins of $H$ (independent of $\Delta$) such that states of $\sigma'_i$ are uniquely identified with configurations of $\sigma_{P_i}$, such that $|E'_{\sigma'} - E_\sigma| \leq \delta$ for any energy level $E_\sigma<\Delta$. We refer to the spins $P=\cup P_i$ in the simulation that correspond to the spins of the target model as the ``physical spins''.
 \label{part:eigenstates}
\item The partition function $Z_H(\beta) = \sum_\sigma e^{-\beta H(\sigma)}$ of $H$ reproduces the partition function $Z_{H'}(\beta) = \sum_{s'} e^{-\beta H'(s')}$ of $H'$ up to constant rescaling, to within arbitrarily small relative error:
\begin{equation}
\left| \frac{Z_{H}(\beta) - \mu Z_{H'}(\beta)}{\mu Z_{H'}(\beta)} \right| \leq  \left(e^{\beta \delta} -1 \right) + O\left(\frac{e^{-\Delta}}{\mu Z_{H'}(\beta)}\right)
\end{equation}
for some known constant $\mu$.
 \label{part:partition_function}
\end{enumerate}
\end{definition}

\noindent In \cite{Cubitt:2016} an efficient universal model was defined as follows:
\begin{definition}[Efficient universal model (definition 4 from \cite{Cubitt:2016})] \label{univ}
We say that a universal model is efficient if, for any Hamiltonian $H' = \sum_{I=1}^mh_I$ on $n$ spins composed of $m$ separate $k$-body terms $h_I$, $H'$ can be simulated by some Hamiltonian $H$ from the model specified by $\text{poly}(m,2^k)$ parameters, and acting on $\text{poly}(n,m,2^k)$ spins.
\end{definition}
\noindent This definition will be too restrictive for the translationally invariant case, but is included here for completeness. A generalisation of this definition which applies to the translationally invariant case will be introduced in Section \ref{proof}.

\subsection{Computing energy levels of $k$-local Hamiltonians}
All of our constructions make use of a Turing machine, which reads in a description of a $k$-local Hamiltonian in binary, and outputs the energy of each $k$-local term in the Hamiltonian for a given spin configuration. We do not explicitly construct the transition rules for such a machine, appealing to the fact that calculating the energy levels of a $k$-local Hamiltonian term to any constant precision is evidently an efficiently computable function, and thus it is possible to construct a Turing machine which does this.

\subsection{Complexity classes}
In complexity theory problems are classified into complexity classes, defined by the amounts of certain computational resources needed to solve the problem.  There are a number of complexity classes which will be relevant for this work.

\begin{definition}[P: Polynomial time]
The class of decision problems solvable in polynomial time by a Turing machine.\footnote{A decision problem is a problem that can be phrased as a `YES / NO' question of the input parameters. Equivalently, it is a function $f:\{0,1\}^* \rightarrow \{0,1\}$.}
\end{definition}

\begin{definition}[NP: Non-deterministic polynomial time]
The class of decision problems decidable in polynomial time by a non-deterministic Turing machine. Equivalently, NP is the class of decision problems for which if the answer is YES then there is a proof, polynomial in the length of the input, that can be verified in P.
\end{definition}

\begin{definition}[EXP: Exponential time]
The class of all decision problems solvable in time $O(2^{p(n)})$ time, where $p(n)$ is a polynomial function of the length of the input, $n$.
\end{definition}

\begin{definition}[NEXP: Non-deterministic exponential time]
The class of decision problems decidable in exponential time by a non-deterministic Turing machine. Equivalently, the class of decision problems for which if the answer is YES then there is a proof, exponential in the length of the input, that can be verified in EXP.
\end{definition}

EXP and NEXP are the exponential time analogues of P and NP respectively.


\begin{definition}[Polynomial time reduction]
Problem $A$ reduces to problem $B$ if there exists a map $f: A \rightarrow B$ such that $b = f(a)$ is a YES instance of $B$ if and only if $a$ is a YES instance of $A$ and the map $f:A \rightarrow B$ is poly-time computable.
\end{definition}
If $A$ reduces to $B$ then we can solve $A$ by transforming it into $B$ and solving $B$. We use the notation $A \leq B$ to denote that $A$ reduces to $B$.

\begin{definition}
A problem $A$ is hard for a complexity class $\mathcal{C}$ if every problem in $\mathcal{C}$ can be reduced to $A$.
\end{definition}
\begin{definition}
A problem $A$ is complete for a complexity class $\mathcal{C}$ if $A$ is hard for $\mathcal{C}$ and $A$ is in $\mathcal{C}$.
\end{definition}
The complete problems for any particular complexity class can be considered the hardest problems in that complexity class.

A decision problem which will be useful in this work is the \textsc{Ground State Energy} problem (GSE):
\begin{definition}[\textsc{Ground State Energy}, Definition 8 from \cite{Cubitt:2016}] \label{gse_def}
\hfill\newline
The ground state energy problem of a model $\mathcal{M} = \{H_\alpha\}$, asks: given $H_\alpha \in \mathcal{M}$ and $c\in\mathbb{Q}$, is there a configuration of spins $\sigma$ such that $H_\alpha(\sigma) \leq c$.
\end{definition}

\section{NEXP-complete tiling construction}\label{NEXP-complete}
In this section we review the NEXP-complete tiling construction originally published in \cite{Gottesman:2009}, on which our universality construction is based. The tiling problem in \cite{Gottesman:2009} is formally defined as follows:
\begin{definition}[\textsc{Tiling}, Definition 2.1 from \cite{Gottesman:2009}] \label{tiling_def}
\hfill\newline
\textbf{Problem parameters:} A set of tiles $T = \{t_1,...,t_m\}$. A set of horizontal constraints $H \subseteq T \times T$ such that if $t_i$ is placed to the left of $t_j$, then it must be the case that $(t_i,t_j) \in H$. A set of vertical constraints $V \subseteq T \times T$ such that if $t_i$ is placed below $t_j$, then it must be the case that $(t_i,t_j) \in V$. A designated tile $t_1$ that must be placed in the four corners of the grid. \\
\textbf{Problem Input:} Integer $N$, specified in binary. \\
\textbf{Output:} Determine whether there is a valid tiling of an $N \times N$ grid.
\end{definition}

In order to demonstrate that this is NEXP complete,\footnote{Although tiling in general was already known to be NEXP-complete, the construction in \cite{Gottesman:2009} is the first reduction where the set of tiles, constraints and boundary conditions are kept fixed, and are not given as part of the input.}  Gottesman and Irani \cite{Gottesman:2009} make use of the fact that tiling is Turing complete \cite{Berger:1966,Robinson}; it is possible to construct sets of tiles and tiling constraints that mimic the behaviour of any Turing machine.

In Section \ref{section_turing} we outline the general idea behind encoding a Turing machine in tiling rules, before going on in Section \ref{gottesman_details} to discuss the specific tiling rules used in \cite{Gottesman:2009} for showing NEXP-completeness when the boundary conditions of the tiling problem specify the tile to be placed in each corner. It is clear that translational invariance is broken at the corners in this construction, but this construction provides the basis for versions of the problem with open and periodic boundary conditions \cite{Gottesman:2009}, which we cover in Sections \ref{open_bc} and \ref{periodic_bc} respectively. We include the tiling rule definitions, and some intuition for why they work, but refer the reader to \cite{Gottesman:2009} for proofs. The notation and tile markings in this section are taken directly from \cite{Gottesman:2009}. Finally in Section \ref{deterministic} we discuss a modification of the NEXP-complete tiling construction which will be useful in our universality proof.

\subsection{Encoding a Turing machine in tiling} \label{section_turing}
Tiling is known to be Turing complete, which means that any Turing machine can be represented by a set of tiling rules \cite{Berger:1966,Robinson}. The basic idea is that it is possible to construct tiling rules such that any row in a valid tiling represents the configuration of the Turing machine tape, internal state, and head position at a particular point in time, and the sequence of rows along the vertical direction represents the sequence of configurations in the time evolution of the Turing machine.

In order to make this explicit, we first review the definition of a Turing machine. We can define a Turing machine as a triple: $M = \langle Q, \Sigma, \delta \rangle$, where $Q$ denotes a non-empty set of states, $\Sigma$ is the Turing machine alphabet, and $\delta: Q \times \Sigma \rightarrow Q \times \Sigma \times \{L,R\}$ is the transition function (L/R denotes that the Turing machine head has moved to the left/right). The Turing machine will have a designated blank symbol $\# \in \Sigma$, starting state $q_0 \in Q$ and final state $q_F \in Q$ \cite{Hopcroft:1979}.

To encode this definition of a Turing machine into a tiling problem, \cite{Gottesman:2009} uses three different varieties of tile. The first variety is denoted $[a]$ where $a \in \Sigma$ is a Turing machine tape symbol. Tiles of this variety denote the symbol on the Turing machine tape away from the position of the Turing machine head. The second tile variety is denoted by a triple $\Sigma \times Q \times \{r,l\}$. These tiles denote the symbol on the Turing machine tape and the state of the head when the head has moved to a location, but not yet acted. The $\{r,l\}$ signify which direction the head came from in its last move. The third tile variety is also denoted by a triple $\Sigma \times Q \times \{R,L\}$. Tiles of this variety denote the state of the head and the tape symbol after the head has acted, and the $\{R,L\}$ denote which direction the head moved~\cite{Gottesman:2009}.

To encode a Turing machine in these tiles, one defines tiling rules that force each pair of adjacent rows to represent a valid update of head position, internal state and tape. The specific tiling rules used in \cite{Gottesman:2009} are given in Section \ref{gottesman_details}. An illustration of a portion of tiles that encodes a Turing machine, reproduced from \cite{Gottesman:2009}, is shown below. In this example the Turing machine is running from bottom to top on the tiling grid, and the Turing machine is undergoing the evolution $(a,q) \rightarrow (b,q',L)$ in the first move, and $(c,q') \rightarrow (f,q'',R)$ in the second move.
\begin{table}[H]
\centering
\begin{tabular}{|c | c | c |}
\hline
$[f, q'', R]$ & $[b,q'',l]$ & $[d]$ \rule[-2em]{0pt}{4.5em}\\
\hline
$[c,q',r]$ & $[b,q',L]$ & $[d]$ \rule[-2em]{0pt}{4.5em}\\
\hline
$[c]$ & $[a,q,r]$ & $[d,q,L]$ \rule[-2em]{0pt}{4.5em}\\
\hline
\end{tabular}
\end{table}
The $[d,q,L]$ in the bottom right corner is from the previous step of the Turing machine evolution. It indicates a Turing machine head which is now in state $q$ and has moved to the left. This is consistent with the $[a,q,r]$ in the bottom middle row which signifies that the state of the Turing machine is $q$, and that the Turing machine tape came from the right. The $[b,q',L]$ shows that the Turing machine is going to execute the step $(a,q) \rightarrow (b,q',L)$, and we can see in the adjacent tile $[c,q',r]$ that indeed the Turing machine head has moved to the left, and its state is now $q'$. Finally in the bottom row $[f, q'', R]$ indicates that the Turing machine is going to execute the step $(c,q') \rightarrow (f,q'',R)$, and the $[b,q'',l]$ tile to the right indicates that indeed the Turing machine tape has moved one step to the right, and is now in state $q''$.

\subsection{NEXP-complete tiling rules with corner tiles fixed}\label{gottesman_details}
We first review the boundary conditions and the resulting tile pattern on the border of the grid in \cite{Gottesman:2009}.

\paragraph{Boundary conditions:} The boundary conditions for the NEXP-complete tiling construction in \cite{Gottesman:2009} require a $\tileC$ to be placed in each of the four corners of the grid. There are then four other tiles which we refer to as boundary tiles $\tileL$, $\tileR$, $\tileN$, $\tileS$. The tiling rules for the boundary tiles are reproduced in Table \ref{boundary}.

\begin{table}
\centering
\begin{minipage}{0.51\textwidth}
\centering
\begin{tabular}{c c c c}
 &  & & Tile on right
 \end{tabular}
\begin{tabular}{c c| c c c c c c}
 &  & $\tileC$ & $\tileL$ & $\tileR$ & $\tileN$ & $\tileS$ & $*$ \\
 \hline
 & $\tileC$ & N & N & N & Y & Y & N \\
 Tile & $\tileL$ & N & N & Y & N & N & \\
 on & $\tileR$ & N & N & N & N & N & N \\
 left & $\tileN$ & Y & N & N & N & N & N \\
  & $\tileS$ & Y & N & N & N & Y & N \\
  & $*$ & N & N & & N & N &
\end{tabular}
\end{minipage}
\centering
\begin{minipage}{0.51\textwidth}
\centering
\begin{tabular}{c c c c}
 &  & & Tile on top
 \end{tabular}
\begin{tabular}{c c| c c c c c c}
 &  & $\tileC$ & $\tileL$ & $\tileR$ & $\tileN$ & $\tileS$ & $*$ \\
 \hline
 & $\tileC$ & N & Y & Y & N & N & N \\
 Tile & $\tileL$ & Y & Y & N & N & N & N  \\
 on & $\tileR$ & Y & N & Y & N & N & N \\
 bottom & $\tileN$ & N & N & N & N & N & N \\
  & $\tileS$ & N & N & N & Y & N &  \\
  & $*$ & N & N & N&  & N &
\end{tabular}
\end{minipage}
\caption{(Table~2 from \cite{Gottesman:2009}) Tiling rules for the boundary tiles. Here a $*$ denotes any interior tile, `N' indicates a disallowed pair of neighbouring tiles, and `Y' denotes an allowed pair of neighbouring tiles. In cases where the rule for the interior tile is not specified this is because the rule depends on the type of interior tile. }
\label{boundary}
\end{table}

These tiling rules enforce that the only place a $\tileL$ tile can go is on the left boundary, the only place a $\tileR$ tile can go is the right boundary, the only place a $\tileS$ tile can go is on the bottom boundary, and the only place a $\tileN$ tile can go is the top boundary. The tiling rules also enforce that the only tiles can go adjacent to a $\tileC$ tile are boundary tiles. If we consider the top boundary only $\tileN$ and $\tileS$ tiles can go to the left or right of a $\tileC$ tile, so the entire top boundary will have to be $\tileN$ or $\tileS$ tiles, but no tile can go below a $\tileS$ tile, so $\tileS$ tiles cannot go on the top boundary. This leaves the entire top boundary, except the corner tiles, as $\tileN$ tiles. Similar logic applies for the other boundaries, and we find that these tiles form a border for a grid, which can be seen in Fig. \ref{pic}.

These boundary tiles allow us to implement special conditions along the borders of the tiling grid, while only breaking translational invariance at the four corners. We will see in Sections \ref{open_bc} and \ref{periodic_bc} how the construction in \cite{Gottesman:2009} is modified to make it fully translationally invariant.

\begin{figure}
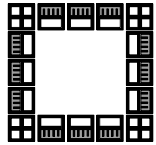

\centering
\begin{tabular}{c@{\extracolsep{0.1em}}c@{}c@{}c@{}c}
\tileC & \tileN & \tileN & \tileN & \tileC \\
\tileL & & & & \tileR \\
\tileL & & & & \tileR \\
\tileL & & & & \tileR \\
\tileC & \tileS & \tileS & \tileS & \tileC
\end{tabular}
\caption{(Figure~1 from~\cite{Gottesman:2009}) The only allowed tiling of the border of an $N = 5$ grid.}
\label{pic}
\end{figure}


\paragraph{Turing machine tiling rules:} The tiling rules required to simulate a Turing machine running from bottom to top on a grid using a tiling grid are given in Table \ref{turing_rules}, with the tile types discussed in Section \ref{section_turing}. The rules enforce that each row in a valid tiling represents the state of the Turing machine tape at a moment in time, and that the vertical evolution in a valid tiling represents the time evolution of the Turing machine tape. (See~\cite{Gottesman:2009} for full details.)

It is convenient to think about the interior of the grid as having two ``layers''. The first layer represents the time-evolution of a binary counter Turing machine running from top to bottom. The second represents that of a non-deterministic Turing machine running from bottom to top. In both cases the main tiling rules are of the form specified in Table \ref{turing_rules}.

It is important to note that there are not actually two layers of tiles in the construction -- we are tiling one grid, using one set of tiles. But, apart from the boundary tiles, each tile in the construction is specified by a pair, denoting its ``layer 1'' type and its ``layer 2'' type: $T = T_1 \times T_2$, and the tiling constraints are given by $H = H_1 \times H_2$ and $V = V_1 \times V_2$. Thinking about these tile markings as representing ``two layers'' of tiles in the interior is convenient, because most of the tiling rules will constrain one of the layers independently of the other.

\begin{table}
\centering
\begin{minipage}{1\textwidth}
\centering
\begin{tabular}{c c c c}
 &  & & Tile on right
 \end{tabular}
\begin{tabular}{c c| c c c c c c}
 & & $[b]$ & $[b,q',r]$ & $[b,q',l]$ & $[b,q',R]$ & $[b,q',L]$ & $\tileR$ \\
 \hline
 &$[a]$  & Y & Y & N & Y & N & Y \\
Tile &$[a,q,r]$  & N & N & N & N & If $q=q'$ & N \\
on &$[a,q,l]$  & Y & N & N & N & N & Y \\
 left & $[a,q,R]$ & N & N & If $q=q'$ & N & N  & N\\
  &$[a,q,L]$  & Y & N & If N & N & N & Y \\
 &$\tileL$  & Y* & Y & If $q'=q_0$ & Y & N & Y
\end{tabular}
\end{minipage}
\centering
\begin{minipage}{1\textwidth}
\centering
\vspace{0.3in}
\begin{tabular}{c c c c}
 &  & & Tile on top
 \end{tabular}
\begin{tabular}{c c| c c c c c }
 & & $[b]$ & $[b,q',r]$ & $[b,q',l]$ & $[b,q',R]$ & $[b,q',L]$  \\
 \hline
 &$[a]$  & If $a=b$ & If $a=b$ & If $a=b$, $q' \neq q_0$ & N & N  \\
Tile &$[a,q,r]$  & N& N & N & If TM rule & If TM rule  \\
on &$[a,q,l]$  & N& N & N & If TM rule & If TM rule  \\
 bottom & $[a,q,R]$ &If $a=b$ & If $a=b$ & If $a=b$, $q' \neq q_0$ & N & N  \\
 &$[a,q,L]$  & If $a=b$ & If $a=b$ & If $a=b$, $q' \neq q_0$ & N & N
\end{tabular}
\end{minipage}
\caption{(Table~3 from~\cite{Gottesman:2009}) Tiling rules for simulating a Turing machine using a tiling, given our boundary conditions. Here `N' indicates a disallowed pair of neighbouring tiles, and `Y' denotes an allowed pair of neighbouring tiles. `If TM rule' means that the tiles can only be placed in that configuration if there is a Turing machine rule allowing that move. The `Y*' entry will be modified later to get the correct starting configuration.}
\label{turing_rules}
\end{table}

\paragraph{First layer of tiling:} For the first layer of the tiling, we will need to implement some additional rules in order to ensure that the binary counter Turing machine begins with the $[\#,q_0,l]$ tile in the top left interior tile, followed by $[\#]$ tiles. This is enforced by additional tiling rules, given in Table \ref{layer1}.

Combining these tiling rules with those for running a Turing machine gives the complete set of tiling rules for the first layer, where the particular Turing machine we are implementing on the first layer is a binary counter Turing machine, $M_{BC}$. It should also be noted that the binary counter Turing machine on the first layer ``runs'' from top to bottom, so the tiling rules for implementing a Turing machine from the previous section will be modified so that the Turing machine ``runs'' in the opposite direction, from bottom to top.

\begin{table}
\centering
 \begin{tabular}{c | c c c c c c }
 Boundary & & & Adjacent interior tile   & & \\
tile & $[a]$ & $[a,q,r]$ & $[a,q,l]$ & $[a,q,R]$ & $[a,q,L]$ \\
\hline
 $\tileS$ & Y & Y & Y & Y & Y \\
$\tileN$ & If $a = \#$ & N & If $a=\#$ and $q=q_0$ & N & N \\
 $\tileL$ & If $a \neq \#$ & Y & If $q=q_0$ & Y & N
\end{tabular}
\caption{(Table~4 from~\cite{Gottesman:2009}) Additional tiling rules for the first layer of tiling.}
\label{layer1}
\end{table}

\paragraph{Second layer of tiling:} For the second layer we would like to copy the output from the binary counter Turing machine to the bottom layer of the second layer, and we would like to enforce that the head of the non-deterministic Turing machine goes at the left-most point of the grid on the bottom row. We also require that in a valid tiling the Turing machine must be in its accepting state, $q_F$, at the top row of the grid. This is achieved with the additional tiling rules given in Table \ref{layer2}.

In order to force the head of the Turing machine to start at the left-most point of the grid, the tiling rules enforce that no symbol from the binary counter Turing machine alphabet can ever go to the right of a $\tileL$ tile. This means that after the Turing machine head has moved on from the leftmost point we will need to overwrite the symbol in the leftmost tile with a new symbol which does not appear in $\Sigma_{M_{BC}}$. This is accomplished by introducing an $a'$ symbol in the alphabet of the non-deterministic Turing machine for every $a \in \Sigma_{M_{BC}}$. Once the Turing machine head has moved on from the leftmost point, it overwrites the symbol on the leftmost tile to its primed version, and this is treated as an ordinary symbol for the remainder of the computation.

With this construction Gottesman and Irani demonstrate that there is a valid tiling of the grid if and only if the non-deterministic Turing machine accepts on input $x$ in $N$ steps, so every problem in NEXP can be reduced to tiling, and therefore tiling is NEXP-complete.

\begin{table}
\centering
 \begin{tabular}{c  | c c c c c }
 Boundary  & &  & Adjacent interior tile  & & \\
 tile   & $[a]$ & $[a,q,r]$ & $[a,q,l]$ & $[a,q,R]$ & $[a,q,L]$ \\
\hline
 $\tileS$ & If $a$ matches layer 1 & N & If $q=q_0$ and $a$ matches layer 1 & N & N \\
 $\tileN$ & Y & If $q=q_F$  & If $q=q_F$ & Y & Y \\
$\tileL$ & If $a \notin \sigma_{M_{BC}}$ & Y & If $q=q_0$ & Y & N
\end{tabular}
\caption{(Table~4 from~\cite{Gottesman:2009}) Additional tiling rules for the second layer of tiling.}
\label{layer2}
\end{table}

\subsection{NEXP-complete weighted tiling with open boundary conditions}\label{open_bc}
To have a NEXP-complete version of the tiling problem with open boundary conditions, one must consider a variant of the tiling problem where the constraints are weighted.

\begin{definition}[Weighted tiling (definition 4.3 from \cite{Gottesman:2009}]
\hfill\newline
\textbf{Problem parameters:} A set of tiles $T = \{t_1,...,t_m\}$. A set of horizontal weights $w_H:T\times T \rightarrow \mathbb{Z}$ such that if $t_i$ is placed to the left of $t_j$ , there is a contribution of $w_H(t_i,t_j)$ to the total cost of the tiling. A set of vertical weights $w_V: T \times T \rightarrow \mathbb{Z}$, such that if $t_i$ is placed below $t_j$, there is a contribution of $w_V(t_i,t_j)$ to the total cost of the tiling. A polynomial $p$. Boundary conditions (a tile to be placed at all four corners, open boundary conditions, or periodic boundary conditions). \\
\textbf{Problem input:} Integer $N$, specified in binary. \\
\textbf{Output:} Determine whether there is a tiling of an $N \times N$ grid such that the total cost is at most $p(N)$.
\end{definition}

In \cite{Gottesman:2009} Gottesman and Irani construct a NEXP-complete weighted tiling problem with open boundary conditions, where $p(N) = -4$. This is done using a three layer tiling construction.\footnote{As before, the `three-layer' terminology is just introduced to make the discussion clearer, in reality there is just one layer of tiles, with each tile specified by a triple $T_1 \times T_2 \times T_3$.} The tile types for layers 1 and 2 are unchanged from Section \ref{gottesman_details}, and there are five tile types which can be used in the third layer: $\tileSR$, $\tileSL$, $\tileNL$, $\tileNR$ and $\tileW$. The tiling weights for the third layer are given in Table \ref{open_bc_rules}. These tiling rules ensure that the minimum cost of tiling the third layer is $-4$. The optimal configuration is shown in Fig. \ref{optimal}.

\begin{table}
\centering
\begin{minipage}{1\textwidth}
\centering
\begin{tabular}{c c c c}
 &  & & Tile on right
 \end{tabular}
 \linebreak
\begin{tabular}{c c| c c c c c }
 & & $\tileNL$ & $\tileNR$ & $\tileSL$ & $\tileSR$ & $\tileW$ \\
 \hline
 &$\tileNL$  & $+4$& +4 & +4 & +4 & -1 \\
Tile &$\tileNR$  & +4 & +4 & +4 & +4 & +2  \\
on &$\tileSL$  & +4 & +4 & +4 & +4 & -1 \\
 left & $\tileSR$ & +4 & +4 & +4& +4 & +2 \\
  &$\tileW$  & +2 & -1 &  +2 & -1 & 0
\end{tabular}
\end{minipage}
\begin{minipage}{1\textwidth}
\centering
\begin{tabular}{c c c c}
 &  & & Tile on right
 \end{tabular}
 \linebreak
\begin{tabular}{c c| c c c c c }
 & & $\tileNL$ & $\tileNR$ & $\tileSL$ & $\tileSR$ & $\tileW$ \\
 \hline
 &$\tileNL$  & $+4$& +4 & +4 & +4 & +2 \\
Tile &$\tileNR$  & +4 & +4 & +4 & +4 & +2  \\
on &$\tileSL$  & +4 & +4 & +4 & +4 & 0 \\
 left & $\tileSR$ & +4 & +4 & +4& +4 & 0 \\
  &$\tileW$  & 0 & 0 &  +2 & +2 & 0
\end{tabular}
\end{minipage}
\caption{(Table~8 from~\cite{Gottesman:2009}) The tiling weights for third layer tiles in \textsc{Weighted Tiling} with open boundary conditions.}
\label{open_bc_rules}
\end{table}

\begin{figure}
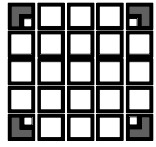

\begin{centering}
\begin{tabular}{c@{\extracolsep{0.1em}}c@{}c@{}c@{}c}
\tileNL & \tileW & \tileW & \tileW & \tileNR \\
\tileW & \tileW & \tileW & \tileW & \tileW \\
\tileW & \tileW & \tileW & \tileW & \tileW \\
\tileW & \tileW & \tileW & \tileW & \tileW \\
\tileSL & \tileW & \tileW & \tileW & \tileSR \\
\end{tabular}
\caption{(Figure~3 from~\cite{Gottesman:2009}) The optimal configuration of the third layer of tiles in \textsc{Weighted Tiling} with open boundary conditions.}
\label{optimal}
\end{centering}
\end{figure}

If we insist that any $\tileSR$, $\tileSL$, $\tileNL$ or $\tileNR$ tile in the third layer must correspond to a $\tileC$ in the main layer then the optimal overall tiling has layers 1 and 2 constrained in precisely the way required for the NEXP-complete construction outlined in Section \ref{gottesman_details}. If we assign all forbidden pairs of neighbouring tiles from the tiling rules in Section \ref{gottesman_details} a weight of +1, then the tiling problem from Section \ref{gottesman_details} reduces to \textsc{Weighted Tiling} with open boundary conditions.

\subsection{NEXP-complete weighted tiling with periodic boundary conditions}\label{periodic_bc}
In \cite{Gottesman:2009} Gottesman and Irani also construct a NEXP-complete \textsc{Weighted Tiling} with periodic boundary conditions, where $p(N)$ is again a constant, in this case $p(N)=+2$. As for open boundary conditions, the periodic boundary condition construction adds a third layer to the two-layer construction from Section \ref{gottesman_details}. The tile types which can be used in the third layer are $\tileC$, $\tileH$,$\tileV$, $\tileW$ and $\tileB$. The tiling rules for the third layer are given in Table \ref{periodic_bc_table}. The minimum cost of the third layer is +2 and one possible optimal configuration is shown in Fig. \ref{optimal_pc}.\footnote{The optimal configuration is now only uniquely defined up to translations, as the row of $\tileH$ tiles and column of $\tileV$ tiles could occur anywhere without changing the overall cost.}

The horizontal and vertical lines in Fig. \ref{optimal_pc} delineate the boundary of an $(N-1)\times (N-1)$ grid. Weighted tiling constraints can be implemented to ensure that in an optimal tiling of the overall grid, layers 1 and 2 are constrained to have $\tileC$ at the border of this grid. Therefore, if we assign all forbidden pairs of neighbouring tiles from the tiling rules in Section \ref{gottesman_details} a weight of +1, the tiling problem from Section \ref{gottesman_details} reduces to \textsc{Weighted Tiling} with periodic boundary conditions.

It should be noted that this construction requires that $N$ be odd. This restriction is unimportant both in the NEXP-completeness result and in our universality constructions, as we are free to choose an $N$ satisfying this constraint. (The restriction to odd $N$ can in fact be lifted using aperiodic tilings, see \cite{Cubitt:2015}.)

\begin{table}
\centering
\begin{minipage}{1\textwidth}
\centering
\begin{tabular}{c c c c}
 &  & & Tile on right
 \end{tabular}
 \linebreak
\begin{tabular}{c c| c c c c c }
 & & $\tileW$ & $\tileB$ & $\tileH$ & $\tileV$ & $\tileC$ \\
 \hline
 &$\tileW$  & $+3$& 0 & +3 & 0 & +3 \\
Tile &$\tileB$  & 0 & +3& +3 & 0 & +3  \\
on &$\tileH$  & +3 & +3 & 0 & +3 & +1 \\
 left & $\tileV$ & 0 & 0& +3 & +3 & +3 \\
  &$\tileC$  & +3 & +3 &  +1 & +3 & +3
\end{tabular}
\end{minipage}
\begin{minipage}{1\textwidth}
\centering
\begin{tabular}{c c c c}
 &  & & Tile on right
 \end{tabular}
 \linebreak
\begin{tabular}{c c| c c c c c }
 & & $\tileW$ & $\tileB$ & $\tileH$ & $\tileV$ & $\tileC$ \\
 \hline
 &$\tileW$  & $+3$& 0 & 0 & +3 & +3 \\
Tile &$\tileB$  & 0 & +3 & 0 & +3 & +3  \\
on &$\tileH$  & 0 & 0 & +3 & +3 & +3 \\
 left & $\tileV$ & +3 & +3 & +3 & 0 & 0 \\
  &$\tileC$  & +3 & +3 &  +3 & 0 & +3
\end{tabular}
\end{minipage}
\caption{(Table~9 from~\cite{Gottesman:2009}) The tiling weights for third layer tiles in \textsc{Weighted Tiling} with periodic boundary conditions.}
\label{periodic_bc_table}
\end{table}

\begin{figure}
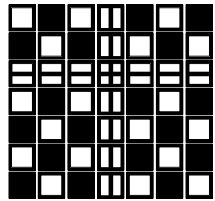

\begin{centering}
\begin{tabular}{c@{\extracolsep{0.1em}}c@{}c@{}c@{}c@{}c@{}c}
\tileW & \tileB & \tileW & \tileV & \tileB& \tileW & \tileB \\
\tileB & \tileW& \tileB & \tileV & \tileW& \tileB& \tileW \\
\tileH & \tileH & \tileH& \tileC & \tileH & \tileH & \tileH \\
\tileW & \tileB & \tileW & \tileV & \tileB& \tileW & \tileB \\
\tileB & \tileW& \tileB & \tileV & \tileW& \tileB& \tileW \\
\tileW & \tileB & \tileW & \tileV & \tileB& \tileW & \tileB \\
\tileB & \tileW& \tileB & \tileV & \tileW& \tileB& \tileW \\
\end{tabular}
\caption{(Modified version of figure~2a from~\cite{Gottesman:2009}) One possible optimal configuration of the third layer of tiles in \textsc{Weighted Tiling} with periodic boundary conditions.}
\label{optimal_pc}
\end{centering}
\end{figure}

\subsection{Using tiling to encode a deterministic Turing machine} \label{deterministic}
In \cite{Gottesman:2009} the tiling rules for the bottom boundary copy the entire bottom row of the first layer to the bottom row of the second layer, and used this as the input to a non-deterministic Turing machine. In their construction it was therefore only the problem instance $x$ which was given as input to the Turing machine.

For our purposes it is more useful to rephrase the construction in terms of a deterministic Turing machine, which takes some additional input $w$ as a witness. In order to do so we need to modify the rules for the bottom boundary in the NEXP-complete tiling construction, so that the part of the boundary that is \emph{not} specifying the problem instance $x$ is allowed to be any Turing machine alphabet tile, rather than being constrained to be the blank symbol as is the case in \cite{Gottesman:2009}. These unconstrained tiles then play the role of the witness -- the verifier Turing machine will accept if these tiles form a valid witness, and reject otherwise.
To achieve this we need to update the tiling rule between the bottom tile and the $[a]$ tile for layer~2. The updated tiling rules are summarised in Table \ref{deterministic-rules}.
\begin{table}
\centering
\begin{tabular}{c | c}
Boundary tile & Adjacent interior tile \\
location & $[a]$ \\
 \hline
 $\tileS$ & If layer 1 is $\#$ or if $a$ matches layer 1
\end{tabular}
\caption{Modified tiling rules to simulate a deterministic Turing machine.}
\label{deterministic-rules}
\end{table}

We now have that at the bottom boundary the problem instance $x$ is copied to the bottom row of the second layer, and the rest of the tiles in the bottom row of the second layer are allowed to be any of the Turing machine alphabet tiles. The rest of the tiling rules are unchanged from the NEXP-complete tiling construction, except that the Turing machine rules which are encoded in the tiling rules will be rules for a deterministic Turing machine, rather than for a non-deterministic Turing machine. This modification to the construction does not change the complexity of the problem. Whilst the alteration to a deterministic Turing machine does mean that each row is uniquely determined by the previous row (except in the case of the bottom boundary), the bottom interior row is not given as input, and the tiling problem remains NEXP-complete.

It should be noted that if the witness takes up the entire remaining length of the bottom row then the Turing machine will not have space to carry out any computation. It will take $N-3$ steps just to read in the input, and after $N-3$ steps the entire grid is tiled. However, if we want the tiling rules to remain translationally invariant we cannot restrict the length of the unconstrained ``witness'' tiles via tiling rules. This will not, however, be an issue in our construction. We will choose the encoding of the problem instance $x$ to include a specification the length of the witness, say $|w|$. The Turing machine will then only read in $|w|$ of the unconstrained tiles as part of its input, and will then begin running the witness verification computation without reading any further input. The remainder of the unconstrained tiles on the bottom row do not form part of the witness, and their states are essentially arbitrary. 

\section{Universality proofs}\label{universality}
\subsection{Tiling constraint Hamiltonians}\label{Tiling constraint Hamiltonians}

The tiling problem is translationally invariant as the tiling rules are the same everywhere on the grid. The Gottesman and Irani construction therefore suggests that translational invariance may not be a barrier to universality. Indeed, it's well known that the existence of a translationally invariant classical Hamiltonian whose GSE problem (Definition \ref{gse_def}) is NEXP-complete follows from \cite{Gottesman:2009}, an argument we encapsulate in the following lemma for later convenience:

\begin{lemma} \label{tiling_constraint}
Every weighted tiling problem can be represented by a classical, translationally invariant, nearest-neighbour spin Hamiltonian on a 2D square lattice, in such a way that tiling configurations correspond to spin configurations, and the energy of any spin configuration is identical to the weight of the corresponding tiling configuration.
\end{lemma}
\begin{proof}
Consider a tiling problem described by a set of tiles $T = \{t\}_{t=1}^{t=d}$, a set of horizontal weights $w_H: T \times T \rightarrow \mathbb{Z}$, and a set of vertical weights $w_V: T \times T \rightarrow \mathbb{Z}$, acting on an $N \times N$ square grid.

Consider a graph $G = (V,E)$ which is restricted to be a $N \times N$ square lattice. Let there be a $d$-dimensional Ising spin, $\sigma_i \in \{t\}_{t=1}^{t=d}$, associated with each vertex $i \in V$. Construct the Hamiltonian $H_T$:
\begin{equation}
H_T(\{\sigma_i\}) =  \sum_{(i,j)\in E_H} w_H(\sigma_i,\sigma_j) + \sum_{(i,j)\in E_V}w_v(\sigma_i,\sigma_j)
\end{equation}
where $E_H$, $E_V$ denote the sets of horizontal and vertical edges respectively.

We can see that there is a one-to-one mapping between the possible spin states of the Ising spins $\sigma_i$ and the tiles in the tiling problem, and that $H_T$ assigns an energy penalty to pairs of adjacent spins which is equal to the weight given to the corresponding pairs of neighbouring tiles.
The Hamiltonian $H_{T}$ is translationally invariant because its local interaction terms are the tiling rules, which are themselves translationally invariant.
\end{proof}
Throughout our proofs we will refer to Hamiltonians derived in this way as `tiling constraint Hamiltonians'. The following rephrases the result from \cite{Gottesman:2009} in terms of classical Hamiltonians:

\begin{corollary} \label{NEXP}
There exists a single, fixed, translationally invariant classical Hamiltonian, $H_{tiling}$, with open boundary conditions, whose GSE problem (taking $c = -4$) is NEXP-complete.\footnote{This corollary applies equally well to translationally invariant Hamiltonians with periodic boundary conditions, taking $c = +2$.}
\end{corollary}


\subsection{Two parameter model} \label{proof}
We can now prove our first result, that there exists a family of translationally invariant classical Hamiltonians with two parameters which can simulate all other classical Hamiltonians. We will prove the result for open boundary conditions, but the construction works equally well for periodic boundary conditions.\footnote{The proof for open boundary conditions requires we set the parameter $c = -4$, and use the third layer tiling rules from Section \ref{open_bc}. For periodic boundary conditions, one should instead use $c = +2$ and the third layer tiling rules from Section \ref{periodic_bc}.} We begin by defining the tiling model:

\begin{definition}[Tiling model]
  A `tiling model' is a family of Hamiltonians specified by a graph $G = (V,E)$ which is restricted to be a 2D square lattice; a tiling constraint Hamiltonian, $H_T$; an energy offset $c$; an energy penalty $\Delta\in\mathbb{R}$; and a ``flag'' energy $\alpha\in\mathbb{R}$.

  A discrete, classical spin $\sigma_i$ which can take values in some finite set $\mathcal{S}$ is associated with each vertex $i \in V$, and Hamiltonians in the model are given by:
\begin{equation}
H(\{\sigma_i\}) = \Delta (H_T(\{\sigma_i\}) - c) + \alpha \sum_{j \in V} f_{{\tilde{\gamma}}}(\sigma_j)
\end{equation}
where ${\tilde{\gamma}} \in \mathcal{S}$ and $f_{\tilde{\gamma}}: \mathcal{S} \rightarrow \mathbb{Z}$ is the function $f_{{\tilde{\gamma}}}({\tilde{\gamma}}) = +1$, and $\forall \sigma_i \neq {\tilde{\gamma}} \;\; f_{{\tilde{\gamma}}}(\sigma_i) = 0$.
\end{definition}


\begin{lemma}[Simulating two-level Ising spins using the tiling model] \label{two-level}
There exists a single fixed tiling constraint Hamiltonian, $H_{univ}$, with open boundary conditions, such that a tiling model with $H_T = H_{univ}$ and $c = -4$ can simulate all $k$-local classical Hamiltonians on two-level Ising spins.
\end{lemma}

\begin{proof}
\paragraph{Proof overview:} In order to prove the lemma, we will explicitly construct $H_{univ}$.
Our construction takes inspiration from the NEXP-complete weighted tiling problem with periodic boundary conditions constructed in \cite{Gottesman:2009}. As in the Gottesman and Irani construction, the problem instance (in this case the Hamiltonian to be simulated, $H'$) is encoded in the binary expansion of $N$, the size of the grid to be tiled. The tiling is carried out in three layers:\footnote{As in \cite{Gottesman:2009} this `three layer' language is used to make explanations of the tiling rules simpler, but in reality we are just tiling one plane, using one set of tiles and tiling rules.} the first layer implements a binary counter Turing machine, so that the bottom row of the first layer contains $N$ in binary -- a description of $H'$.

This description of $H'$ is copied to the bottom layer of the second row (taking up approximately $\log(N)$ tiles), and the rest of the tiles in the bottom row of the second layer are left unconstrained. The first $n$ of the unconstrained tiles act as the `physical' spins (see Definition \ref{simulation_definition}). The tiling rules for the second layer encode a Turing machine, which reads in the description of $H'$ and the configuration of the physical spins, and computes the energy that should be assigned to this configuration in multiples of $\alpha$. The tiling rules then force the correct number of spins to be in the ${\tilde{\gamma}}$ state, so that the system has the desired energy. The third layer is used to implement boundary conditions on the first two layers.

Any non-optimal tiling will incur an energy penalty of $\Delta$, so provided $\Delta > E_{\max}(H')$ any spin configuration which corresponds to an non-optimal tiling will not have energy below the energy cut-off. As required by Definition \ref{simulation_definition}, this energy cut-off can be made arbitrarily large by increasing the $\Delta$ parameter in the construction. (Similarly to the non-translationally-invariant models constructed in \cite{Cubitt:2016}.)

\paragraph{Encoding $H'$:}
The Hamiltonian to be simulated, $H'$, is encoded in the binary expansion of $N$, the size of the grid to be tiled. Let $H' = \sum_{I=1}^mh_I$ be a $k$-local Hamiltonian containing $m$ terms, acting on $n$ two-dimensional spins. $H'$ is then fully described by $n$ and the energy levels associated with each of the $2^k$ possible spin configurations for each of the $m$ $k$-local terms.

Since we only need to simulate energies up to some fixed precision $\delta$, we can equivalently assume without loss of generality that each $h_I$ is a function $h_I:\mathcal{S}^{\times k} \rightarrow \mathbb{Q}$.
The set $E = \cup_{I=1}^m \text{spec}(h_I)$ (where  $\text{spec}(h_I)$ denotes the spectrum of $h_I$) is then a set of rational numbers, and we can define $\alpha$ to be the greatest common measure of $E$.
A complete description of $H'$ then consists of $n$ and, for each of the $m$ $k$-local terms, $2^k$ parameters giving the energy contribution (in integer multiples of $\alpha$) of that local term for each of the $2^k$ possible configurations of the $k$ spins it acts on.

It is worth considering explicitly how the Hamiltonian to be simulated, $H'$, will be encoded in binary. The first thing to note is that since we are in binary we have only two symbols available to us. As the length of the description depends on the Hamiltonian, we need to use a self-delimiting code. (An alternative would be to encode in ternary, and reserve one symbol exclusively as a delimeter. But using a self-delimiting code and sticking to binary makes the tiling construction slightly simpler.)

We will use a simple self-delimiting code known as the Elias-$\gamma$ code \cite{Fenwick:2003}. The length of an integer $z$ encoded in the Elias-$\gamma$ code scales as $O(\log(z))$, so there is only a constant overhead compared with the binary representation of $z$. If we denote the binary representation of $z$ by $B(z)$, and its length by $|B(z)|$, then the Elias-$\gamma$ code for $z$ is given by $|B(z)|-1$ zeros that indicate how long the input will be, followed by $B(z)$ itself. All binary numbers (with the exception of zero) begin with a one, so $B(z)$ starts with a one, and this delimits $B(z)$ from the $|B(z)|-1$ zeros. The Elias-$\gamma$ code does not code zero or negative integers. But we do not need negative integers for our construction, and we can easily include zero in the code by adding one to every number before coding, and subtracting one after decoding. We will refer to this as the Elias-$\gamma'$ coding.

Having established how to encode individual parameters in binary, we can now consider how to encode the full Hamiltonian. Every value needed to specify the Hamiltonian will be represented in Elias-$\gamma'$ coding. For the purpose of the encoding we label the $n$ spins in the original system by integers $i = 1:n$, corresponding to the order in which these are represented in the physical spins that act as input to the Turing machine. The encoding of the Hamiltonian begins with $n$, followed by $k$. Each of the $m$ $k$-local terms in $H'$ is specified by giving the label of each spin involved in that interaction (a total of $k$ integers), followed by the energy values for each of the $2^k$ possible spin configurations of the $k$ spins acted on by that local term, ordered by the index of the configuration considered as a binary number. These energy values are specified by the closest integer multiple of $\alpha$ to the desired value. ($l$-local terms for $l<k$ can be specified by arbitrarily picking $k-l$ spins to pad the number of spins to $k$, but specifying identical energy values for configurations differing only on those extra $l-k$ spins).

After the last of the $k$-local terms has been described in this way, there will be a single~1, delimiting the end of the description of $H'$. This will be unambiguous, as it is the only time that an integer in Elias-$\gamma'$ coding will be followed by a one, rather than by a string of zeros.\footnote{The 1 at the end also ensures that $N$ is an odd number. This is not important for the open boundary conditions, but would be necessary if we were using periodic boundary conditions (see Section \ref{periodic_bc}). (Note that for using periodic boundary conditions, the size of the grid to be tiled is actually $N+4$.)} If we let $\gamma'(z)$ denote the integer $z$ represented in Elias-$\gamma'$ code then putting all of this together, we can represent the Hamiltonian as follows:
\begin{equation}
H' := \gamma'(n)\cdot \gamma'(k)\cdot \left[ \gamma'(i)^{\cdot k} \cdot \left( \gamma'(\lambda_j) \right)^{\cdot 2^k} \right]^{\cdot m}\cdot 1
\end{equation}
where $\cdot$ denotes concatenation of bit strings, the $i$ denotes the index of the spins in each $k$-local term, and the $\lambda_j$ indicate the energy levels in terms of multiples of $\alpha$ of each $k$-local term in the Hamiltonian.

This encoding can unambiguously specify any $k$-local Hamiltonian in (the binary representation of) an integer $N$. It will take $O(\log(n))$ bits to specify $n$ in Elias-$\gamma'$ code, $O(\log(k))$ bits to specify $k$ in Elias-$\gamma'$ code, and $O(m2^k)$ bits to specify the energy levels of each $k$-local term. The total number of bits required therefore scales as $O(m2^k) = O(n^k)$, where we have used the fact that without loss of generality, $m \leq {n\choose k} = O(n^k)$ for $k$-local Hamiltonians. The length of $N$ in binary is approximately $\log(N)$, so we have that $\log(N) = O(n^k)$, and hence $N = O(2^{n^k})$.
$N$ is clearly efficiently computable from any other reasonable description of $H'$. 

\paragraph{Tiling rules:}\footnote{The tiling rules in this section are stated in terms of forbidden and allowed pairs of neighbouring tiles. As in Section \ref{open_bc}, all forbidden pairs of neighbouring tiles have weight +1, and all allowed pairs have weight 0.}
The tiling rules for the third layer of tiles and the border of layers one and two are unchanged from the construction in \cite{Gottesman:2009}, outlined in Sections \ref{open_bc} and \ref{gottesman_details} respectively.
%
The first layer of the tiling will implement a binary counter Turing machine from top to bottom, as before. So that the bottom row of the first layer in any optimal tiling contains $N$ in binary, with the remaining tiles blank. 
For simplicity we will assume that the tape alphabet for the binary counter Turing machine is given by $\Sigma_{M_{BC}} = \{0,1,\#\}$, where $\#$ denotes the blank symbol. The tiling rules needed to implement a binary counter Turing machine on the first layer of the tiling are unchanged from \cite{Gottesman:2009}.

\paragraph{Bottom boundary:}
The tiling rules for the second layer will encode a Turing machine which reads in the description of the Hamiltonian to be simulated in binary, and the state of the physical spins, and calculates the energy of the system being simulated. For concreteness we will assume that the tiling rules encode a Turing machine with six alphabet symbols: $\Sigma_{M_U} = \{0,1,0',1',\gamma, \# \}$. The $\#$ denotes a blank symbol, the 0,$0'$ and 1,$1'$ are used during the computation, and the $\gamma$ is only used in the output of a computation. The input to the Turing machine represents the state of the physical spins. As the spins being simulated are two-level we need to update the tiling rules so that only $0$ and $1$ symbols can appear as input. The updated rules are provided in Table \ref{lower-boundary}.

\begin{table}
\centering
\begin{tabular}{c | c}
Boundary tile & Adjacent interior tile \\
location & $[a]$ \\
 \hline
 $\tileS$ & If layer 1 is $\#$ and $a = 0,1$ OR if layer 1 is not $\#$ and $a$ matches layer 1
\end{tabular}
\caption{Modified tiling rules for the bottom boundary in the universality construction.}
\label{lower-boundary}
\end{table}
These rules are combined with the rules from \cite{Gottesman:2009}, which state that in layer 2 an $[a]$ tile can only go next to the left boundary if $a\notin \Sigma_{M_{BC}}$. This ensures that the left-most position contains the head of the Turing machine, and that the output of the first layer (i.e. $N$ in binary) is copied to the first approximately $\log(N)$ tiles, while the remaining tiles are constrained to be the tiles representing the Turing alphabet states $0$ or $1$.  After the Turing machine head has moved on from the left-most position we need to have the left-most tile in an alphabet state which is not in $\Sigma_{M_{BC}}$, so the alphabet symbol from the left-most tile is updated to its primed version, which obeys the tiling rule, and is treated as a $0,1$ for the rest of the computation.

\paragraph{Computation:}
With these tiling rules for the bottom boundary, the only optimal tilings of the bottom row will have the Turing machine head in the left-most position, with $N$ in binary copied to the first approximately $\log(N)$ tiles, and the remaining tiles constrained to be alphabet tiles containing either the symbols $0$ or $1$.

The Turing machine reads in the program, which contains a description of the Hamiltonian to be simulated, including the number of spins the Hamiltonian is acting on, $n$. It will then read in a further $n$ input bits, which determine the state of the $n$ `physical' spins in the simulation. The Turing machine will then calculate the energy (in terms of multiples of $\alpha$) arising from each of the $k$-body terms in the Hamiltonian given the state of the physical spins, and will output $0$ if a $k$-body term contributes no energy with the given spin configuration, or $a$ copies of the symbol $\gamma$ if a $k$-body term contributes energy $a \alpha$ with the given spin configuration. After each of the $m$ $k$-local terms has been calculated, the final state of the Turing machine tape will be a string of $0$s and $\gamma$s containing $M$ $\gamma$s (where the total energy of the particular spin configuration is $M\alpha$), with a number of $0$ and $1$s remaining on parts of the tape that have not been used for output.

We have seen that $N = O(2^{n^k})$. It will take the Turing machine $O(\log(N) + n) = O(n^k)$ steps to read in the input. The remainder of the calculation simply involves determining which of the $2^k$ possible configurations the relevant $k$ spins are in, for each of the $k$-local terms, and outputting the corresponding energy. For each $k$-local term this can clearly be done in time $O(\text{poly}(2^k))$, giving a time for the total computation of $O(\text{poly}(n^k))$. The system has $N = O(2^{n^k})$ rows of computation available to it, so will finish the computation in the space available in the tiling grid.

At the point the Turing machine has finished its computation, and determined the energy contribution from each $k$-body term in the Hamiltonian, the Turing machine tape will contain a string of $0$s and $\gamma$s, followed by some 0s and 1s left over from the initial input, and from the part of the tape that didn't form part of the initial input.

Note that there will always be enough room to write all the $\gamma$'s in a single $N$-length row of tiles. The description of $H$, which is the binary expansion of $N$, lists the energy levels of each of the $k$-local terms in terms of multiples of $\alpha$, written in binary, which means that $N$ always grows faster than the required number of $\gamma$'s.

\paragraph{Top boundary:}
The tiling rules introduced so far ensure that once the computation finishes the row which encodes the final line in the computation will just be copied until the grid is full. However, this means that there will be more than $M$ of the $\gamma$ symbols in the full grid, so if we give energy to the $\gamma$ symbols we will not end up with the correct energy. To circumvent this issue, we introduce a new type of tile, which we will denote $\langle\tilde{\gamma}\rangle$, and we will introduce tiling rules which ensure that the $\langle\tilde{\gamma}\rangle$ tiles only appear in the top interior row of the tiling, so we can give energy to just these tiles without introducing terms in the Hamiltonian which break translational invariance.

The additional tiling rules associated with the $\langle\tilde{\gamma}\rangle$ tile are given in Table \ref{additional-rules}.
\begin{table}
\centering
\begin{tabular}{c | c c c c }
 &   & Position of $\langle\tilde{\gamma}\rangle$ & &\\
Adjacent tile type & Below & Above & Left & Right \\
 \hline
$\tileN$ & Y & N & N & N \\
$[a]$ & N & If $a = \gamma$ & Y & Y \\
$[a,q,r]$ & N & If $a = \gamma$ & Y & Y \\
$[a,q,l]$ & N & If $a = \gamma$ & Y & Y \\
$[a,q,R]$ & N & If $a = \gamma$& Y & Y \\
$[a,q,L]$ & N & If $a = \gamma$ & Y & Y \\
$\tileL$ & N & N & N & Y \\
$\tileR$ & N & N & Y & N \\
$\tileS$ & N & N & N & N \\
$\langle\tilde{\gamma}\rangle$ & N & N & Y & Y
\end{tabular}
\caption{Additional tiling rules for the $\langle \tilde{\gamma} \rangle$ tile.}
\label{additional-rules}
\end{table}
In an optimal tiling the $\langle\tilde{\gamma}\rangle$ tile can only appear in the top interior row as the only tile that can appear directly above a $\langle\tilde{\gamma}\rangle$ tile is a top boundary tile. In order to ensure that every $\gamma$ symbol that is output in the computation is copied to a $\langle\tilde{\gamma}\rangle$ tile in the top interior row we need to update the rules for the top boundary. The updated rules are provided in Table \ref{top-two}.

\begin{table}
\centering
Adjacent interior tile
\begin{tabular}{c | c c c c c c}
Boundary tile & $[a]$ & $[a,q,r]$ & $[a,q,l]$  & $[a,q,R]$ & $[a,q,L]$ & $\langle\tilde{\gamma}\rangle$  \\
 \hline
 $\tileN$ & If $a \neq \gamma$ & $a \neq \gamma$ & $a \neq \gamma$ & $a \neq \gamma$ & $a \neq \gamma$ & Y
 \end{tabular}
 \caption{Modified tiling rules for the top boundary in the two-parameter universality construction.}
 \label{top-two}
\end{table}
These rules ensure that in every optimal tiling 
there will be exactly $M$ of the $\langle\tilde{\gamma}\rangle$ tiles in the top interior row, and none appearing  elsewhere in the tiling. If we associate the $\langle\tilde{\gamma}\rangle$ tile with the ${\tilde{\gamma}}$ state of the spins, then the energy of this system is precisely $M\alpha$, the energy of the physical spins being simulated.

\paragraph{Partition function:} We have now established that for all optimal tilings, the energy of the simulator system will be equal to the energy of the target system when the spin configuration of the target system is equal to the spin configuration of the physical spins, and for all non-optimal tilings the energy of the simulator system will be greater than $\Delta$, so $H$ meets the first two requirements to be considered a simulation of $H'$. The final requirement to consider is that $H$ reproduces the partition function of $H'$ to within arbitrarily small relative error.

For this, it is not enough that $H$ has the same energy levels as $H'$. It must also introduce the same additional degeneracy to each of the energy levels of the original Hamiltonian (a proof of this can be found below). In our construction, the spins that encode $H'$ are uniquely determined for all optimal tilings, so introduce no additional degeneracy to any energy levels below $\Delta$.

As we are encoding a deterministic Turing machine, the tiles in each row in an optimal tiling are uniquely determined by the tiles in the preceeding row, so likewise in optimal tilings the spins in the higher rows do not introduce any additional degeneracy. The only spins we need to consider are therefore the spins that correspond to the tiles in the bottom interior row that are constrained to be the alphabet tiles $0$ or $1$, but that are not part of the set of physical spins. The states of these tiles are entirely independent of the states of the physical tiles. As such, they will introduce the same degeneracy to all possible configurations of physical spins, and hence to all the energy levels. The partition function of $H$ therefore reproduces the partition function of $H'$ to within arbitrarily small relative error up to physically irrelevant rescaling.

This completes the proof of Lemma \ref{two-level}.
\end{proof}

\begin{lemma}
Let $H$ be a Hamiltonian on $N$ $d$-level Ising spins that approximates all energy levels of $H'$ to within error $\delta$, up to an energy cut-off $\Delta$. Then
\begin{equation}
\left| \frac{Z_{H}(\beta) - \mu Z_{H'}(\beta)}{\mu Z_{H'}(\beta)} \right| \leq \left(e^{\beta \delta} -1 \right) + O\left(\frac{e^{-\Delta}}{\mu Z_{H'}(\beta)}\right)
\end{equation}
(cf.\ Definition \ref{simulation_definition}) if and only if $H$ introduces the same additional degeneracy $\mu$ to each of the energy levels of $H'$.
\end{lemma}
\begin{proof}
Denote the spin configurations of $H$ and $H'$ by $\sigma$  and $\sigma'$, respectively. Assume that $H$ simulates $H'$, and the degeneracy of each energy level of the $H$ system is multiplied by a factor of $\mu$ compared with the corresponding energy level of the $H'$ system. Then the relative error in the partition function is given by:
\begin{equation}
\begin{split}
\left| \frac{Z_{H}(\beta) - \mu Z_{H'}(\beta)}{\mu Z_{H'}(\beta)} \right| & = \left| \frac{ \sum_{\sigma}e^{-\beta H(\sigma)} - \mu \sum_{\sigma'}e^{-\beta H'(\sigma')}}{ \mu \sum_{\sigma'}e^{-\beta H'(\sigma')}} \right| \\
& = \left| \frac{  \sum_{\sigma:H(\sigma) < \Delta} e^{-\beta H(\sigma)} + \sum_{\sigma:H(\sigma) \geq \Delta} e^{-\beta H(\sigma)} - \mu \sum_{\sigma'}e^{-\beta H'(\sigma')}}{ \mu \sum_{\sigma'}e^{-\beta H'(\sigma')}} \right| \\
&  \leq \left| \frac{ \mu \sum_{\sigma'} e^{-\beta \left(H'(\sigma') + \delta_{\sigma}\right)} - \mu \sum_{\sigma'}e^{-\beta H'(\sigma')}}{ \mu \sum_{\sigma'}e^{-\beta H'(\sigma')}} + \frac{\sum_{\sigma:H(\sigma) \geq \Delta}e^{-\beta \Delta}}{ \mu Z_{H'}(\beta)} \right|  \\
&  \leq \left| \frac{ e^{\beta \delta} \mu  \sum_{\sigma'} e^{-\beta H'(\sigma')} - \mu \sum_{\sigma'}e^{-\beta H'(\sigma')}}{ \mu \sum_{\sigma'}e^{-\beta H'(\sigma')}} + \frac{ d^{N} e^{-\beta \Delta}}{\mu Z_{H'}(\beta)} \right| \\
& =  \left(e^{\beta \delta} -1 \right)  + O\left(\frac{e^{-\Delta}}{\mu Z_{H'}(\beta)}\right)
\end{split}
\end{equation}
In the second step we have used the fact that below $\Delta$ the energy levels of $H$ are within $\delta$ of the energy levels of $H'$ but repeated $\mu$ times, in order to replace the sum over $\sigma$ configurations with a sum over $\sigma'$ configurations. The $\delta_{\sigma'}$ can take on values in the interval $[-\delta,+\delta]$. In the penultimate step we have assumed the worst case (all $\delta_{\sigma'} = -\delta$) to upper-bound the error. \par

In order to see the only if direction, consider the case where $H$ approximates all energy levels of $H'$ to within $\delta$ below energy cut-off $\Delta$, but where it does not introduce the same degeneracy to each energy level. Denote the degeneracy of spin configuration $\sigma'$ by $m_{\sigma'}$. Then the relative error in the partition function is given by:
\begin{equation}
\begin{split}
\left| \frac{Z_{H}(\beta) - \mu Z_{H'}(\beta)}{\mu Z_{H'}(\beta)} \right| & = \left| \frac{\sum_{\sigma}e^{-\beta H(\sigma)} - \mu \sum_{\sigma'}e^{-\beta H'(\sigma')}}{ \mu \sum_{\sigma'}e^{-\beta H'(\sigma')}} \right| \\
& =\left| \frac{ \sum_{\sigma:H(\sigma) < \Delta} e^{-\beta H(\sigma)} + \sum_{\sigma:H(\sigma) \geq \Delta} e^{-\beta H(\sigma)} - \mu \sum_{\sigma'}e^{-\beta H'(\sigma')}}{ \mu \sum_{\sigma'}e^{-\beta H'(\sigma')}} \right| \\
& =\frac{\left| \sum_{\sigma'} m_{\sigma'} e^{-\beta \left( H'(\sigma') + \delta_{\sigma} \right)} - \mu \sum_{\sigma'}e^{-\beta H'(\sigma')}\right|}{ \mu \sum_{\sigma'}e^{-\beta H'(\sigma')} } + O\left(\frac{e^{-\Delta}}{\mu Z_{H'}(\beta)}\right).
\end{split}
\end{equation}
The second term is independent of $\delta$, and the first term will only be upper-bounded by $e^{\beta \delta} -1$ if
\begin{equation} \label{linearly_independent}
\left| \sum_{\sigma'} e^{-\beta H(\sigma')} \left(m_{\sigma'}e^{-\beta \delta_{\sigma'}} - \mu \right) \right| - \mu(e^{\beta \delta} - 1 )\sum_{\sigma'}e^{-\beta H'(\sigma')} \leq 0.
\end{equation}

\noindent There are two cases to consider:\\
\paragraph{Case 1:} $\sum_{\sigma'} e^{-\beta H(\sigma')} \left(m_{\sigma'}e^{-\beta \delta_{\sigma'}} - \mu \right) < 0$. \\
Eq. \ref{linearly_independent} becomes:
\begin{equation} \label{case_1}
\sum_{\sigma'}e^{-\beta H'(\sigma')}\left[ \left(\mu - m_{\sigma'} e^{-\beta \delta_{\sigma'}} \right) +\mu \left(1 - e^{\beta \delta} \right) \right] \leq 0.
\end{equation}
This must hold for all $\beta > 0$, and all $\delta \in (0,1)$. Taking the limit $\beta \delta \rightarrow 0$, Eq. \ref{case_1} becomes:
\begin{equation} \label{not_possible}
\sum_{\sigma'}e^{-\beta H'(\sigma')} \left(\mu - m_{\sigma'} e^{-\beta \delta_{\sigma'}} \right) \leq 0.
\end{equation}
But $\sum_{\sigma'} e^{-\beta H(\sigma'))} \left(\mu - m_{\sigma'}e^{-\beta \delta_{\sigma'}}\right) > 0$ by assumption, so it is not possible to satisfy Eq. \ref{not_possible}, hence Eq. \ref{linearly_independent} cannot be satisfied for sufficiently small $\delta$.\\

\noindent\paragraph{Case 2:} $\sum_{\sigma'} e^{-\beta H(\sigma')} \left(m_{\sigma'}e^{-\beta \delta_{\sigma'}} - \mu \right) \geq 0$. \\
In this case, we can rewrite Eq. \ref{linearly_independent} as:
\begin{equation} \label{case_2_a}
\sum_{\sigma'}e^{-\beta H'(\sigma')}\left[ e^{-\beta \delta_{\sigma'}} \left(m_{\sigma'}  - \mu  \right) + \mu \left(e^{-\beta \delta_{\sigma'}} - e^{\beta \delta} \right)  \right] \leq 0.
\end{equation}
Take the limit $\beta \delta \rightarrow 0$ (which implies $\forall \sigma' \; \beta \delta_{\sigma'} \rightarrow 0$). In this limit $\left(e^{-\beta \delta_{\sigma'}} - e^{\beta \delta} \right)  \rightarrow 0$ and Eq. \ref{case_2_a} becomes:
\begin{equation} \label{case_2_b}
\sum_{\sigma'}e^{-\beta \left( H'(\sigma') + \delta_{\sigma'} \right)} \left(m_{\sigma'}  - \mu  \right) \leq 0.
\end{equation}

Now, $\sum_{\sigma'} e^{-\beta H(\sigma'))} \left(m_{\sigma'}e^{-\beta \delta_{\sigma'}} - \mu \right) \geq 0$ by assupmtion. We can rewrite this as:
\begin{equation} \label{case_2_c}
\sum_{\sigma'} e^{-\beta \left( H'(\sigma') + \delta_{\sigma'} \right)} \left[  \left(m_{\sigma'}  - \mu  \right) +  \mu e^{\beta \delta_{\sigma'}} \left(e^{-\beta \delta_{\sigma'}}  - 1\right) \right]  \geq 0
\end{equation}
We know that $e^{-\beta \left( H'(\sigma') + \delta_{\sigma'} \right)} \geq 0$ and $m_{\sigma'}  - \mu \in \mathbb{Z}$. By taking the limit $\beta \delta_{\sigma'} \rightarrow 0$ we can make $\mu e^{\beta \delta_{\sigma'}} \left(e^{-\beta \delta_{\sigma'}} - 1 \right)$ arbitrarily small, without the first term vanishing. In particular we can choose  $\mu e^{\beta \delta_{\sigma'}} \left(1 - e^{-\beta \delta_{\sigma'}} \right) \ll 1$. The only way to satisfy Eq. \ref{case_2_c} is then to have $m_{\sigma'}  - \mu \in \mathbb{N}$, which implies:
\begin{equation}
\sum_{\sigma'}e^{-\beta \left( H'(\sigma') + \delta_{\sigma'} \right)} \left(m_{\sigma'}  - \mu  \right) \geq 0.
\end{equation}
Together with Eq. \ref{case_2_b}, this implies:
\begin{equation} \label{case_2_d}
\sum_{\sigma'}e^{-\beta \left( H'(\sigma') + \delta_{\sigma'} \right)} \left(m_{\sigma'}  - \mu  \right) = 0
\end{equation}

We can assume without loss of generality that the $H'(\sigma')+\delta_\sigma$ are all distinct.\footnote{If there exist $\sigma_i'$ and $\sigma_j'$ such that $H'(\sigma_i') + \delta_{\sigma_i'} = H'(\sigma_j') + \delta_{\sigma_j'}$ then we can combine them in a single term with coefficient $m_{\sigma_i'} + m_{\sigma_j'} - 2\mu$.} The functions $e^{-\beta (H'(\sigma') + \delta_{\sigma})}$ are then linearly independent, so the only solution to Eq. \ref{case_2_c}, and therefore also the only solution to Eq. \ref{linearly_independent} is $\forall\sigma'\; m_{\sigma'} = \mu$, as claimed.
\end{proof}

It should be noted that in our proof of Lemma \ref{two-level} we have implicitly assumed that all the energy levels in the target Hamiltonian $H'$ are positive. This is a valid assumption as any classical Hamiltonian on a finite number of spins can be made positive by adding a constant energy shift. We could also prove the more general result where we allow $H'$ to have negative energy levels, but this requires including negative energy terms in our universal model, which introduces additional complexity to the construction, as we will see in Section \ref{two-parameters} when we prove our main result: the one-parameter universal model.

It is instructive to consider the number of parameters, and number of spins, required in the construction. The number of parameters needed to specify $H'$ is $O(\text{poly}(m,2^k))$, but the number of parameters in the simulator Hamiltonian $H$ is just two: $\alpha, \Delta$, together with the size of the lattice the Hamiltonian is acting on, $N$. We have already calculated that $N = O(2^{n^k})$, so the number of spins in the simulation is given by $N^2 = O(2^{n^k})$. Comparing this with the definition of universality given in the introduction, we see that the number of parameters needed to specify $H$ is in keeping with this definition of universality, but that the number of spins needed for the simulation is not efficient according to that definition.

However, this definition of efficiency is too restrictive for the translationally invariant case. A translationally invariant Hamiltonian on $N$ spins can be described using only $\text{poly}(\log(N))$ bits of information, whereas a $k$-local Hamiltonian which breaks translational invariance in general requires $\text{poly}(N)$ bits of information. The number of bits required to describe a $k$-local Hamiltonian on $n$ spins is $O(n^k)$. So by a simple counting argument we can see that it is not possible to encode all the information about a $k$-local Hamiltonian on $n$ spins in any translationally invariant Hamiltonian on $\text{poly}(n,m,2^k)$ spins, as this would require encoding $O(n^k)$ bits of information in $O(\log(n))$ bits. As such, we extend the definition of efficiency for universal models in such a way that it coincides with the original definition in \cite{Cubitt:2016} for the models considered previously, but is also meaningful for translationally invariant models:

\begin{definition}[Efficient universal model] \label{univ_2}
We say that a universal model is efficient if, for any Hamiltonian $H' = \sum_{I=1}^mh_I$ on $n$ spins composed of $m$ separate $k$-body terms $h_I$, $H'$ can be simulated by some Hamiltonian $H$ from the model specified by $\text{poly}(m,2^k)$ parameters which can be computed in time $\text{poly}(n,m,2^k)$, described by $\text{poly}(n,m,2^k)$ bits of information.
\end{definition}

This clearly encompasses the previous definition for non-translationally invariant Hamiltonians, as for these Hamiltonians the number of bits needed to describe the Hamiltonian scales polynomially with the number of spins in the system. But it also encompasses the possibility of efficient, universal, translationally invariant models.

With this modified definition of efficiency we can state and prove our main result for the two-parameter case:

\begin{theorem}[Translationally invariant universal model] \label{main}
The tiling model with $H_T = H_{univ}$, $c = -4$ and open boundary conditions forms a translationally invariant, efficient, universal model.
\end{theorem}
\begin{proof}
In Lemma \ref{two-level} we showed that any any $k$-local Hamiltonian on two-level Ising spins can be simulated by a Hamiltonian from this tiling model. The 2D Ising model with fields meets this requirements, and was already shown to be a universal model in \cite{Cubitt:2016}. Since we can simulate a universal model, our tiling model is itself universal.

If $H' = \sum_{I=1}^mh_I$ is a $k$-local Hamiltonian acting on $n$ spins, then the Ising model $H_{Ising}$ simulating it will be specified by $\text{poly}(m,2^k)$ parameters and act on $\text{poly}(n,m,2^k)$ spins.
Thus the tiling model Hamiltonian $H$ simulating this will be specified by two parameters together with the lattice size, and will act on $O(2^{\text{poly}(n,m,2^k)})$ spins, so will be described by $O(\text{poly}(n,m,2^k))$ bits of information.
We demonstrated in the proof of Lemma \ref{two-level} that the size of the lattice needed to simulate $H'$ can be computed in time $O(n^k)$.
\end{proof}

It should be noted that an alternative method to prove Theorem \ref{main} would be to generalise our construction from the proof of Lemma \ref{two-level} to deal with $d$-level spins, and show directly that this tiling model can simulate any Hamiltonian on discrete Ising spins. This would require modifying the mapping between the states of the physical spins on the simulator system, so that instead of one physical spin on the simulator system mapping to one spin on the original system, we would now have $\log_2(d)$ physical spins (which we will refer to as a physical set) on the simulator system mapping to one spin on the original system. In the case where $\log_2(d)$ is an integer we would still have a one-to-one mapping between the states of the physical spins and the states of the original spins. In the case where $\log_2(d)$ is not an integer the number of spins in a physical set would be rounded up to the nearest integer, and some states of the physical set would not correspond to any state of the original spin. In order to handle instances where the states of the physical set do not correspond to any state of the original spin we would introduce an additional term in the Hamiltonian, $\sum_{i \in V}\Delta f_{\nu}(\sigma_i)$, where $f_\nu: \mathcal{S} \rightarrow \mathbb{Z}$ is given by $f_\nu(\nu) = 1$, $f_\nu(\sigma_i) = 0$ $ \forall \sigma_i \neq \nu$ and $\nu$ is a spin state that corresponds to a new alphabet symbol, $\nu$, in the layer 2 Turing machine. We could construct Turing machine rules in such a way that if the states of the physical set do not correspond to a valid configuration the Turing machine will output a $\nu$ symbol, picking up energy $\Delta$ and therefore taking that configuration above the energy cut-off. This does not affect the number of parameters in the model, as we were already using the $\Delta$ parameter in the construction.

The argument sketched above gives an alternative, direct proof that this tiling model could simulate all Hamiltonians on discrete spins. In order to show that this tiling model can simulate all Hamiltonians on continuous spins which are Lipschitz-continuous in each argument we would then follow the argument in \cite{Cubitt:2016}, showing that they can be simulated to arbitrary accuracy using discrete spins. The encoding of $H'$ in the binary expansion of $N$ would also need to specify the value of $d$.

Calculating the local spin dimension of our translationally invariant universal spin model would require explicitly constructing the transition rules for the Turing machine in layer 2 of the tiling. However, since there exist universal Turing machines with 3 alphabet symbols and 9 states we can upper bound the local spin dimension by assuming that a (9,3) Turing machine was used.\footnote{Although the Turing machine used in layer 2 has 6 alphabet symbols, 3 of these symbols are used for specific simulation purposes, and cannot be used in the computation.}  Doing so we find that the local spin dimension of the translationally invariant universal Hamiltonian is 43,510. We have made no effort to optimise the local spin dimension in this work - our focus was on proving that translationally invariant universal models existed. In previous Hamiltonian complexity work the local dimension from the quantum construction in \cite{Gottesman:2009} was brought down from $O(10^{6})$ to 42 in \cite{Bausch:2017} by encoding simpler models of computation in the Hamiltonian. Similar techniques could be applied to this construction to reduce the local spin dimension.

\subsection{One parameter model} \label{two-parameters}
Having established the existence of a universal model which requires only two parameters ($\Delta, \alpha$) it is natural to ask whether all these parameters are strictly necessary, or whether it is possible to generate a universal model requiring fewer parameters. In order to show that it is indeed possible to construct a universal model with only one parameter we require a technical lemma regarding irrational rotations on the unit circle.

We begin by introducing some notation. The unit circle will be denoted $S^1$, with the points on $S^1$ represented by the real interval $[0,1]$. A rotation of a point $x \in S^1$ about an angle $\theta$ is then given by $R_\theta(x) = \left(x + \frac{\theta}{2\pi} \right)\mbox{mod }1$. The orbit of a point $x$ under a rotation $R_\theta$ is the set of points $O = \{R^n_\theta(x) | n \in \mathbb{Z} \}$, where $R^n_\theta(x)$ indicates that the rotation was applied $n$ times.

\begin{lemma} \label{rotations}
  The orbit of any point on the circle under an irrational rotation $R_\theta(x)$, where $\theta = 2\pi \alpha$, $\alpha \notin \mathbb{Q}$ 
  is dense in $S^1$.
\end{lemma}
\begin{proof}
This is a standard result with a number of proofs. We sketch one argument here.

First consider two points in the orbit, $R^n_\theta(x)$ and $R^m_\theta(x)$. Using the definition of the rotation we can write:
\begin{equation}
R^n_\theta(x) = \left(x + n\frac{\theta}{2\pi} \right)\mbox{mod }1 = \left(x + n\alpha \right)\mbox{mod }1
\end{equation}
Similarly:
\begin{equation}
R^m_\theta(x) = \left(x + m\alpha \right)\mbox{mod }1
\end{equation}
Hence if two points in the orbit are equal, we must have:
\begin{equation}
\left(x + n\alpha \right)\mbox{mod }1 = \left(x + m\alpha \right)\mbox{mod }1
\end{equation}
Simplifying, this gives:
\begin{equation}
\alpha = \frac{b-a}{n-m}
\end{equation}
where $a,b \in \mathbb{Z}$ such that $\left(x + n\alpha \right)\mbox{mod }1 = x + n\alpha - a$, and $\left(x + m\alpha \right)\mbox{mod }1 = x + m\alpha - b$. If both $a$ and $b$ were equal to zero, or $a=b$, then $\alpha$ would be equal to zero, which contradicts our initial statement. Hence at least one of $a$ or $b$ does not equal zero, and $a \neq b$. We have, therefore, that either $n=m$ or that $\alpha$ is a ratio of integers, which contradicts our assertion that $\alpha \notin \mathbb{Q}$. Hence $R^n_\theta(x) \neq R^m_\theta(x)$  for $n \neq m$.

We now apply the pigeonhole principle, which states that if we try to fit $n$ items into $m$ containers, where $n > m$, then at least one container will contain more than one item. Consider dividing the unit circle into $Z \in \mathbb{Z}$ disjoint intervals of length $\epsilon = \frac{1}{Z}$. We have just shown that every point on an irrational orbit is distinct, so the points $R^0_\theta(x), R^1_\theta(x),...,R^{Z}_\theta(x)$ are all distinct points on $S^1$. We have, therefore, $Z+1$ distinct points on $S^1$, but only $Z$ disjoint intervals, so at least one of our intervals contains two points. Let $R^l_\theta(x)$ and $R^m_\theta(x)$ be two such points which fall into the same interval, where $0 < m < l < Z$.

We have that the intervals are of length $\epsilon$, so the distance $|R^l_\theta(x) - R^m_\theta(x)|< \epsilon$. Rotation preserves distances, hence $|R^{l-m}_\theta(x) - x| < \epsilon$. The orbit $O = \{R^{n(l-m)}_\theta(x) | n \in \mathbb{Z} \}$ is therefore $\epsilon$-dense in $S^1$, since the first $(Z+1)$-points on the orbit cover $S^1$ in equidistant steps which are separated by less than $\epsilon$.

We have left $x$ and $Z$ arbitrary, and can make $\epsilon$ arbitrarily small by increasing $Z$, hence the orbit of any point on the circle under an irrational rotation is dense in $S^1$.
\end{proof}

It is useful to consider how many points on the orbit we have to take to be $\epsilon$ close to every point on the unit circle. From the proof we see that the first $(Z+1)$-points of the orbit $O = \{R^{n(l-m)}_\theta(x) | n \in \mathbb{Z} \}$, where $l-m \leq Z$ cover $S^1$ in equidistant steps which are separated by less than $\epsilon$. Hence to be $\epsilon$ close to every point in the circle requires at most $(l-m)(Z+1) = O(Z^2) = O(\frac{1}{\epsilon^2})$ points on the orbit. In fact, Weyl's equidistribution theorem demonstrates that the points on an irrational orbit are equidistributed on the unit circle, so we can do slightly better, and need just $O(\frac{1}{\epsilon})$ points on the orbit. The proof is omitted here but can be found in Chapter 12 of \cite{Miller:2006}.

\begin{corollary}
The set of real numbers given by: $T = \{a\sqrt{2} - b |a, b \in \{0, \mathbb{Z}^+\}\}$ is dense in $\mathbb{R}$.
\end{corollary}
\begin{proof}
Lemma \ref{rotations} implies that $T$ is dense in $[0,1]$. Any point in $\mathbb{R}$ can be reached from $[0,1]$ by adding or subtracting integers. The definition of $T$ allows you to subtract multiples of 1 so we can subtract arbitrary integers, and 1 is in the set $[0,1]$. We can also add multiples of a number which is arbitrarily close to 1. Thus, for any $\epsilon$, we can get $\epsilon$-close to any $z \in \mathbb{Z}^+$ by adding together $z$ copies of a number which is $\frac{\epsilon}{z}$-close to 1.
\end{proof}

It should be noted that using $\sqrt{2}$ is arbitrary -- any irrational number would do. We can now demonstrate that there exists a translationally invariant model on one parameter which is universal. As in the two-parameter case we are assuming open boundary conditions, but the construction can easily be adapted to work equally well with periodic boundary conditions. We first define the model:

\begin{definition}[Reduced parameter tiling model]
  A `reduced parameter tiling model' is a family of Hamiltonians specified by a graph $G = (V,E)$ which is restricted to be a 2D square lattice, a tiling constraint Hamiltonian $H_T$; an energy offset $c$; and an energy penalty, $\tilde{\Delta}$.

  A discrete, classical spin $\sigma_i$ which can take on values in the set $\mathcal{S}$ is associated with each vertex $i \in V$, and Hamiltonians in the model are given by:
\begin{equation}
H(\{\sigma_i\}) = \tilde{\Delta} (H_T(\{\sigma_i\}) - c) + \sum_{j \in V} f_{\tilde{\gamma},\tilde{\nu}}(\sigma_j)
\end{equation}
where $\tilde{\gamma}, \tilde{\eta} \in \mathcal{S}$ and $f_{\tilde{\gamma},\tilde{\nu}}: \mathcal{S} \rightarrow \mathbb{Z}$ is the function $f_{\tilde{\gamma},\tilde{\nu}}(\tilde{\gamma}) = +\sqrt{2}$, $f_{\tilde{\gamma},\tilde{\nu}}(\tilde{\nu}) = -1$ and $\forall \sigma_i \neq \tilde{\nu}, \tilde{\gamma}\;\; f_{\tilde{\gamma},\tilde{\nu}}(\sigma_i) = 0$.
\end{definition}


 \begin{lemma}[Simulating two-level Ising spins using the reduced parameter tiling model] \label{two-level-reduced}
There exists a single fixed tiling constraint Hamiltonian, $H_{R}$, with open boundary conditions, such that the reduced parameter tiling model with $H_T = H_{R}$, and $c = -4$, can simulate all k-local classical Hamiltonians on two-level Ising spins.
\end{lemma}
\begin{proof}
The proof of this lemma is closely related to the proof of Lemma \ref{two-level}, so we omit some of the details here, and refer back to the previous proof.

\paragraph{Outline of proof:}
The target Hamiltonian $H'$ is again encoded in the size of the grid to be tiled, $N$. As before, the first layer of the tiling implements a binary counter Turing machine from top to bottom, so that the bottom row of the first layer contains $N$ in binary; this is copied to the bottom row of the second layer, taking up approximately $\log(N)$ tiles, and the remaining tiles are unconstrained, with a fixed subset of these tiles acting as the `physical' spins. The second layer of tiles again encodes a Turing machine, which reads in the description of $H'$ ($N$ in binary), and the state of the physical spins, and outputs the energy that should assigned to this configuration by computing the energy of each $k$-local term to the desired precision. However, this time the energy is computed in the form $a\sqrt{2} - b$. The tiling rules then force the appropriate number of spins to be in the $\tilde{\gamma}$ and $\tilde{\eta}$ states, so the system has the correct energy. The third layer of tiles is used to implement boundary conditions.

Any non-optimal tiling will incur an energy penalty $\tilde{\Delta}$, but we now have negative energy terms in the Hamiltonian,\footnote{There are already negative energy terms arising from the third layer of tiling, but the analysis in \cite{Gottesman:2009} demonstrates that the number of energy bonuses from these negative terms is bounded by $c=-4$.} so a non-optimal tiling may also pick up some energy bonuses (i.e. negative energy contributions). The maximum number of energy bonuses any non-optimal tiling could pick up is $N^2$. We will later discuss ways to bound the number of tiling errors needed to induce a certain number of energy bonuses. But for now we assume that it is possible to cause $N^2$ energy bonuses with $O(1)$ tiling errors. So in order to push all incorrect tilings above energy $\Delta$ we require $\tilde{\Delta} = \Delta + N^2$. With this condition any non-optimal tiling will have an energy greater than $\Delta$, so provided $\Delta > E_{\max(H')}$ any spin configuration which corresponds to an non-optimal tiling will be above the energy cut-off.

\paragraph{Encoding $H'$:}
As in the previous construction the Hamiltonian to be simulated $H'$ is encoded in the binary expansion of $N$.
Let $H' = \sum_{I=1}^m h_I$ be a $k$-local Hamiltonian acting on $n$ two-dimensional Ising spins. As before, $H'$ is fully described by the energy of each of the $2^k$ possible configurations of each of the $m$ $k$-local terms. Each energy will be specified by giving two numbers, $a$ and $b$, where the energy $\lambda = a\sqrt{2} - b$. Let us assume we want to specify each energy to precision $\epsilon$. Then $a$, $b$ scale as $\frac{1}{\epsilon}$. A complete description of $H'$ therefore consists of $n$, and for each of the $m$ $k$-local terms, $2^k$ elements giving the energy, where each of these elements contains two binary numbers which scale as $\frac{1}{\epsilon}$.

The description of the Hamiltonian will again be given in Elias-$\gamma'$ coding, using the same format as in Lemma \ref{two-level}, with the modification that the energy levels will now each be specified by two Elias-$\gamma'$ numbers. This leads to the Hamiltonian $H'$ being specified in binary in the form:
\begin{equation}
H' := \gamma'(n)\cdot \gamma'(k)\cdot \left[ \gamma'(i)^{\cdot k} \cdot \left( \gamma'(a_j) \cdot \gamma'(b_j) \right)^{\cdot 2^k} \right]^{\cdot m}\cdot 1
\end{equation}

In the proof of Lemma \ref{two-level} we found that $N$ scaled as $O(2^{n^k})$. Now that we are including the precision information in $N$ too this scaling will become $O(2^{\frac{2^k}{\epsilon}})$. As in the previous construction, $N$ can clearly be efficiently computed from any other reasonable description of $H'$.

\paragraph{Tiling rules:}\footnote{The tiling rules in this section are stated in terms of forbidden and allowed pairs of neighbouring tiles. As in Section \ref{open_bc}, all forbidden pairs of neighbouring tiles have weight +1, and all allowed pairs have weight 0.}
The tiling rules for the third layer of tiles and the border of layers one and two are unchanged from the construction in \cite{Gottesman:2009}, outlined in Section \ref{open_bc} and Section \ref{gottesman_details} respectively.
The rules for the first layer of the tiling are unchanged from Lemma \ref{two-level}.

\paragraph{Bottom boundary:} As before the second layer will encode a Turing machine. But for this construction we will increase the alphabet of the Turing machine, and use the alphabet $\Sigma_{M_U} = \{0,1,0',1',\gamma,\eta,\#\}$. The meanings of the primed symbols are unchanged from Lemma \ref{two-level}, and the $\gamma$ and $\eta$ are only used in the output. The boundary rules for the bottom boundary are unchanged from Lemma \ref{two-level}, and again ensure that the physical spins can only be in the 0 or 1 state.

\paragraph{Computation:} The computation is essentially unchanged from the previous construction, with the only difference being that now for each $k$-local term in the Hamiltonian the computation will output $a$ $\gamma$s, and $b$ $\eta$s where the energy of that $k$-local term with the particular spin configuration of physical spins is given by $\lambda = a\sqrt{2} - b$. The Turing machine can also output $0$ in the case of a zero energy term. At the point the Turing machine has finished its computation its tape will contain a string of $0$s, $\gamma$s and $\eta$s, followed by some string of 0s and 1s, such that if the total energy of the configuration of physical spins is $E = a'\sqrt{2} - b'$ then the number of $\gamma$s is $a'$ and the number of $\eta$s is $b'$. As before, there will always be enough room to write all the $\gamma$'s and $\eta$'s in a single $N$-length row of tiles.

\paragraph{Top boundary:} The tiling rules for the top boundary are very similar to those in the construction of Lemma \ref{two-level}, but they have to be updated to allow for the inclusion of a $\langle \tilde{\eta}\rangle$ tile. The updated rules for the top boundary are provided in Table \ref{top-one}. The additional tiling rules which ensure that the $\langle \tilde{\eta}\rangle$ tile does not appear elsewhere in an optimal tiling are summarised in Table \ref{additional-eta}.

These tiling rules are just straightforward generalisations of the tiling rules from Lemma \ref{two-level}, with the tiling rules for $\langle \tilde{\eta}\rangle$ tiles following straightforwardly from the tiling rules for $\langle \tilde{\gamma}\rangle$ tiles.
\begin{table}
\centering
\begin{tabular}{c | c c c c }
 &   & Position of $\langle \tilde{\eta}\rangle$ & &\\
Adjacent tile type & Below & Above & Left & Right \\
 \hline
$\tileN$ & Y & N & N & N \\
$[a]$ & N & If $a = \eta$ & Y & Y \\
$[a,q,r]$ & N & If $a = \eta$ & Y & Y \\
$[a,q,l]$ & N & If $a = \eta$ & Y & Y \\
$[a,q,R]$ & N & If $a = \eta$& Y & Y \\
$[a,q,L]$ & N & If $a = \eta$ & Y & Y \\
$\tileL$ & N & N & N & Y \\
$\tileR$ & N & N & Y & N \\
$\tileS$ & N & N & N & N \\
$\langle \tilde{\gamma}\rangle$ & N & N & Y & Y \\
$\langle \tilde{\eta}\rangle$ & N & N & Y & Y
\end{tabular}
\caption{Additional tiling rules for the $\langle \tilde{\eta} \rangle$ tile.}
\label{additional-eta}
\end{table}

\begin{table}
\centering
Adjacent interior tile
\begin{tabular}{c | c c c c c c c}
Boundary tile & $[a]$ & $[a,q,r]$ & $[a,q,l]$  & $[a,q,R]$ & $[a,q,L]$ & $\langle \tilde{\gamma}\rangle$ &  $\langle \tilde{\eta}\rangle$ \\
 \hline
 $\tileN$ & If $a \neq \gamma, \eta$ & $a \neq \gamma, \eta$ & $a \neq \gamma, \eta$ & $a \neq \gamma,\eta$ & $a \neq \gamma,\eta$ & Y & Y
 \end{tabular}
 \caption{Modified tiling rules for the top boundary in the one-parameter universality construction.}
 \label{top-one}
\end{table}

\paragraph{Partition function:}
The argument for reproducing the partition function is unchanged from Lemma \ref{two-level}.

\end{proof}

It should be noted that in this construction we can simulate negative energy levels without a problem, so there is no need to add a constant energy shift to the target Hamiltonian in order to make the target Hamiltonian positive semidefinite as we did in the two-parameter case.
\begin{theorem}[Translationally invariant universal model] \label{main-reduced}
The reduced parameter tiling model in with $H_T = H_{R}$, $c=-4$ and open boundary conditions forms a translationally invariant universal model.
\end{theorem}
\begin{proof}
This proof follows exactly the same argument as Theorem \ref{main}.
\end{proof}

As in the two parameter case, we can upper bound the local spin dimension of the universal Hamiltonian by assuming that a universal Turing machine is used in layer 2 of the tiling. This gives an upper bound of 50,920. As before, the problem of optimising the construction to reduce the local spin dimension is left to future work.

We have now established the existence of a translationally invariant, universal family of Hamiltonians which require only one parameter. On the other hand, our construction now has a number of simulator spins that scales as $O(2^{\frac{n^k}{\epsilon}})$, whereas before it scaled as $O(2^{n^k})$, and the energy penalty in the Hamiltonian now requires an offset that scales as $O(N^2)$. The scaling of $N$ with $\epsilon$ is inevitable, as removing the $\alpha$ parameter in the Hamiltonian means that the precision information has to be encoded in $N$. However, it is not immediately clear whether this scaling of the $\tilde{\Delta}$ term is necessary.

We first note that it is not possible to have a universal model without a precision parameter unless the Hamiltonian contains `energy bonuses':

\begin{lemma}
Let $H^{(x)} = \sum_i h^{(x)}$ be a translationally invariant family of Hamiltonians, where $x$ denotes the parameters specifying Hamiltonians in the model and the sum is over all spins in the simulator system. If $H^{(x)}$ is universal and contains no precision parameter, then the spectrum of $h^{(x)}$ contains at least one energy level which is below that corresponding to the ground state energy on the target systems.
\end{lemma}
\begin{proof}
We can assume without loss of generality that the energy level of the universal Hamiltonian that maps to the ground state energy of the target Hamiltonian is the zero point energy of the universal Hamiltonian, as this can always be achieved by a constant energy shift to the universal Hamiltonian. The lemma can therefore be rephrased as stating that the energy spectrum of $h^{(x)}$, denoted $\text{spec}(h)$, contains at least one negative value.

Let us assume that there exists a translationally invariant universal model where $\text{spec}(h)$ contains only positive values. Now the local Hamiltonian $h^{(x)}$ must act on some fixed number of fixed dimension spins, so the set $\text{spec}(h) \setminus \{0\}$ is finite. Any finite set of real numbers contains a minimum, so $\text{spec}(h) \setminus \{0\}$ contains a minimum, which we will denote $\lambda_{\min}$.

If we consider the overall Hamiltonian $H^{(x)} = \sum_i h^{(x)}$ acting on some large set of simulator spins then the spectrum of $H^{(x)}$, denoted $\text{spec}(H)$, can only contain different combinations of the energy levels of $h(x)$ added together. So $\text{spec}(H) \setminus \{0\}$ must have a minimum at least as large as the minimum of $\text{spec}(h) \setminus \{0\}$ (in general it will be larger). Hence $H^{(x)}$ cannot simulate any Hamiltonian with a non-zero energy level lower then $\lambda_{\min}$. Whichever $\lambda_{\min}$ we choose we can construct a Hamiltonian with non-zero energy lower than $\lambda_{\min}$ that the $H^{(x)}$  cannot simulate, hence $H^{(x)}$ is not universal.
\end{proof}

We have established that any two parameter translationally invariant model without a precision parameter requires `energy bonuses' in the Hamiltonian. This suggests that some offset to $\tilde{\Delta}$ will always be necessary in order to ensure that non-optimal tilings are pushed to energy above $\Delta$.

Reducing the size of the offset would be possible if the computation encoded in the tiling rules can be carried out ``fault tolerantly'', so that the number of energy bonuses incurred by making a single tiling error somewhere in the lattice is bounded. There does exist a construction for a 1-dimensional fault tolerant cellular automata \cite{Gacs:2001}, which could be encoded into the tiling rules, suggesting the possibility of carrying out the computation fault tolerantly and reducing the offset to $\tilde{\Delta}$. However, the construction involves a complicated encoding scheme, and it is not obvious that it would be possible to decode the output of the computation fault tolerantly, so an error in the decoding procedure could still lead to an unbounded number of errors in the output of the simulator, which would give $O(N)$ energy bonuses in the final output. There also exists a 3-dimensional tolerant cellular automata which fault tolerantly simulates a Turing machine; this could be encoded into tiling rules on a 4-dimensional grid, and it requires no decoding \cite{Gacs:1988}. This simpler fault tolerant device seems more likely to provide a route to reducing the offset to $\tilde{\Delta}$, with the caveat that it would only apply to 4-dimensional lattices.

Whether this, or another approach, can be successfully applied to produce a universal model requiring only one parameters with a better scaling of the $\tilde{\Delta}$ offset, is an open question which we leave to future research. However it should be noted that in the worst case $E_{\max}(H')$ can scale as $O(N) = O(2^n)$. In this case, even if the scaling of $\tilde{\Delta}$ with $N$ were improved, we would require it to scale as $O(2^n)$. So our construction already achieves the optimal worst-case asymptotic scaling of $\tilde{\Delta}$ with $n$.

\subsection{Necessity of one parameter to specify a universal model} \label{no-parameters}
We will now demonstrate that the one-parameter universal model is optimal in terms of the number of parameters required to describe a universal model.

\begin{theorem}
It is not possible to construct a translationally invariant universal model in which the number of spins the Hamiltonian acts on is the only variable.
\end{theorem}
\begin{proof}
Assume that there exists a translationally invariant Hamiltonian, $H_{U}$, which takes no parameters other than the number of spins it is acting on, and which is universal. Write $H = \sum_i h$ where $h$ is some $k$-local Hamiltonian containing $m$ terms, acting on $d$-dimensional Ising spins, and where the sum is over all $S = N \times N$ spins on the simulator system. Denote the spectrum of $h$ by $\text{spec}(h) = \{\mu_i\}$. Let $\mu_{max} = \max\{\mu_i\}$ and $\mu_{min} = \min\{\mu_i\}$, where by the minimum we mean the largest magnitude negative value in the set.\footnote{As we saw in the previous section any universal model without a precision parameter will need at least one negative value in its spectrum.}

Consider a target Hamiltonian $H'$, where the spectrum of $H'$ is given by $\text{spec}(H') = \{\lambda_i\}$. As $H_U$ is a universal Hamiltonian it must be able to simulate $H'$ up to arbitrary precision $\epsilon$, and below arbitrary energy cut-off $\Delta$. We saw in the previous section that it is possible to encode the precision into $N$, now we also have to encode the energy cut-off into $N$.

Let us assume that we \emph{can} encode $\Delta$ into $N$ in some way, and that when $H$ acts on $S = N \times N$ spins it simulates $H'$ below energy $\Delta$. Consider a configuration, $\sigma$, of the $S$ spins which is below the energy cut-off, so that $E(\sigma) = \lambda_x \in \{\lambda_i\}$. Now take one spin in the configuration, and change its state to one of the other $d-1$ possible spin states. The maximum energy difference such a change in state can have is bounded by $\mu_{max} - \mu_{min}$.

Now consider simulating $H'$ with a different energy cut-off $\Delta'$, where $\Delta' > \Delta$. We must have that this simulation involves $H$ acting on $S' = N' \times N'$ spins, where $N' \neq N$ for sufficiently larger $\Delta'$. Consider again a configuration $\sigma'$ of the spins, such that $E(\sigma') = \lambda_x$. We can again consider taking one spin in the configuration and changing its state, and the maximum energy difference such a change in state can have will again be bounded by $\mu_{max} - \mu_{min}$. We see that if we are in a spin configuration that is below the energy cut-off, then the energy difference between the current spin configuration and one which differs from the current configuration by the state of just one spin is bounded, and in particular is independent of $N$. Therefore, this energy difference cannot be made arbitrarily large by altering $N$. Hence in order to satisfy Definition \ref{simulation_definition} of simulation, the new spin configuration we reach by altering the state of one spin must also be below the energy cut-off.

We have not restricted $\lambda_x$ in any way, hence this argument must hold for the spin configuration reached via one spin flip from any state below the energy cut-off. So if we start in a configuration with energy $\lambda_x \in \{\lambda_i\}$, after one spin flip we are in another configuration with energy $\lambda_y \in \{\lambda_i\}$, and if we carry out another single spin flip we will again be in a configuration below the energy cut-off.\footnote{In general the spin flip may take us to another configuration below the energy cut-off, or leave us in the same configuration, the important point is that either way the resulting configuration is still below the energy cut-off.} Now, all of the $d^S$ possible spin configurations can be reached by starting in the configuration $\sigma$, and carrying out a series of single-spin flips. Each single-spin flip takes us to a different configuration below the energy cut-off, so we never exceed the energy cut-off. Hence we find that every possible configuration of the $S$ spins must map to one of the energy levels of $H'$.

This argument holds regardless of whether we were trying to achieve the energy cut-off of $\Delta$ or $\Delta'$. Thus we find that whenever $H_U$ is simulating a Hamiltonian $H'$, every spin configuration of the spins in the simulator system must map to one of the energy levels of $H'$.

However, in order to correctly simulate $H'$ it is not enough to reproduce the energy levels. We must also introduce the same additional degeneracy to each energy level, so that the partition function is reproduced. The number of possible spin configurations of $S$ simulator spins is $d^S$. If we denote the number of spin configurations of the original spins by $M$, the requirement of introducing the same additional degeneracy to each energy level of the original system enforces that $\forall M$, $\exists S$ such that $\frac{d^S}{M} \in \mathbb{Z}$.

To see that this isn't possible, consider arbitrary $d$, and select $M$ such that $d$ and $M$ are coprime, i.e. so that the greatest common factor of $d$ and $M$ is 1.\footnote{For any $d$ such an $M$ can always be found. Any prime number greater than $d$ will be coprime with $d$, and there exists an infinite number of prime numbers.} If we let the prime factors of $d$ (ignoring multiplicities) be the set $P = \{p_i\}$, and the prime factors of $M$ be the set $P' = \{p'_j\}$, then from the fact that $d$ and $M$ are coprime we have that $P \cap P' = \emptyset$. Now consider the prime factors of $d^S$. Ignoring multiplicities, the set of prime factors of $d^S$ is the same as the set of prime factors of $d$, so we find that $M$ and $d^S$ have no common factors greater than 1, and are coprime. Hence $\frac{d^S}{M}$ is an irreducible fraction, and in particular $\frac{d^S}{M} \notin \mathbb{Z}$. Therefore for any $d$ we choose there will exist some Hamiltonians which cannot be simulated. Hence our one-parameter, translationally invariant Hamiltonian cannot be universal.
\end{proof}

It should be noted that the requirement that the energy cut-off $\Delta$ can be made arbitrarily large is crucial in order to ensure that the partition function can be approximated arbitrarily well. If we allowed a definition of simulation in which the energy cut-off was fixed we would only be able to reproduce the partition function to within some constant error.

\subsection{A translationally invariant model on $\text{poly}(n)$ spins}
Finally, we note that it would be possible to construct a universal model on $\text{poly}(n)$ spins by allowing the Turing machine that is encoded in the tiling rules to depend on what Hamiltonian is being simulated, which would vastly increase the number of parameters in the model.

In both the two-parameter and one-parameter universal constructions the computation only takes time $\text{poly}(n)$ to run. The reason we need a grid of size $O(2^{n^k})$ is to allow the binary counter Turing machine to provide a description of the Hamiltonian being simulated. If instead of a universal Turing machine we used the tiling rules to encode a Turing machine which specifically computes the energy levels of the target Hamiltonian, then there would be no need for the binary counter Turing machine, and we could simply run the non-universal Turing machine on $\text{poly}(n)$ spins.

The number of parameters needed to specify such a universal model would be large. In particular, since there are only a finite number of Turing machines with a given alphabet, but an infinite number of Hamiltonians we might want to simulate, the alphabets of the Turing machines, and hence the local state space of Hamiltonians in the model, would not be fixed.

\section{Discussion} \label{discussion}

Our main result is that there exists a translationally invariant Hamiltonian which can replicate all classical spin physics, just by tuning one parameter in the Hamiltonian, and varying the number of spins the Hamiltonian is acting on. We mentioned in Section \ref{main_results} that this result has some implications for complexity theory, but perhaps the more interesting implications are those for classical many-body physics.

It was commented on in \cite{Cubitt:2016} that the existence of universal models implies that the properties of classical many-body systems are not determined solely by the symmetries or number of spatial dimensions of the system, nor by the structure of their interaction graph. However, all the universal models constructed in \cite{Cubitt:2016} required inhomogeneous coupling strengths, and it was suggested that the inhomogeneity of the couplings could account for this. We can now go further, and state that even for translationally invariant models with homogenous couplings the properties of the system are not determined by the dimensionality or symmetry of the system, or by the interaction structure, since all physical properties of e.g. 3D spin systems with complicated interaction graphs can be simulated on a 2D square lattice with nearest-neighbour, translationally invariant interactions.

The price that we have paid for translational invariance is that the local dimension of the spins is higher than in the non-translationally invariant case. In the non-translationally invariant case well studied models such as the 2D Ising model with fields were shown to be universal, whereas in the translationally invariant case we have demonstrated universality by explicitly constructing universal models, which have not previously been studied.

It should be noted that in the classical setting physically simulating one spin model using another is unlikely to be of much practical use, as numerical simulations will usually be more efficient. As such the significance of this classical result lies more in what it tells us about the properties of classical spin models, than in technological applications of the result.

When dealing with quantum systems, however, we can no longer efficiently simulate them numerically using a classical computer. Quantum systems can be efficiently numerically simulated using a quantum computer \cite{Lloyd:1996}, however this requires a large-scale, fault-tolerant quantum computer, which is beyond the reach of current technology. This has led to interest in the concept of analogue quantum Hamiltonian simulation -- simulating a quantum system by engineering the Hamiltonian and measuring its properties. There has been experimental progress in designing quantum simulators using a range of implementations \cite{Houck:2012}, \cite{Blatt:2012}, \cite{Bloch:2012}; and theoretical work has provided a rigorous framework for analogue Hamiltonian simulation \cite{Cubitt:2017}.

Within the rigorous framework of analogue quantum Hamiltonian simulation defined in \cite{Cubitt:2017}, there exist quantum spin models that are universal, in the sense that they can simulate the entire physics of any other quantum many-body system. As in the classical case, some of the universal spin models are surprisingly simple, with examples including the Heisenberg, and XY-interaction \cite{Cubitt:2017}. However, every known universal quantum model breaks translational invariance. It would be interesting to consider whether there exists a quantum analogue to results in this paper -- is it possible to construct a universal, translationally invariant quantum model?

In \cite{Gottesman:2009} it was demonstrated that the hardness of calculating the ground state energy of a one-dimensional translationally invariant quantum system is QMA$_{\scaleto{\mbox{EXP}}{5pt}}$-complete, suggesting that in the quantum regime translational invariance is not a barrier to complexity. This work may prove a route to extending our result to the quantum case. But it should be noted that the quantum notion of simulation is more involved than the classical case, and it is not clear that there exists a straightforward generalisation of our result to the quantum case.

\begin{acknowledgements}
TK is supported by the EPSRC Centre for Doctoral Training in Delivering Quantum Technologies [EP/L015242/1].
TC is supported by the Royal Society.
This work was supported by the EPSRC Prosperity Partnership in Quantum Software for Simulation and Modelling (EP/S005021/1).
\end{acknowledgements}

\bibliographystyle{spmpsci}      
\bibliography{References}   

\end{document}